\documentclass[10pt,a4paper]{article}

\usepackage{amsmath,amsgen,latexsym}
\usepackage{amstext,amssymb,amsfonts,latexsym}
\usepackage{graphicx}
\usepackage{theorem}
\usepackage{pifont}

 \newcommand{\bs}{\bigskip}
 \newcommand{\ms}{\medskip}
 \newcommand{\n}{\noindent}
 \newcommand{\s}{\smallskip}
 \newcommand{\hs}[1]{\hspace*{ #1 mm}}
 \newcommand{\vs}[1]{\vspace*{ #1 mm}}

 \newcommand{\setempty}{\varnothing}
 
 \newcommand{\nat}{\mathbb{N}}
 
 \newcommand{\integer}{\mathbb{Z}}

 \newcommand{\PP}{{\cal P}}

\theoremstyle{plain}
\theoremheaderfont{\bfseries}
 \newtheorem{theorem}{Theorem}[section]
 \newtheorem{lemma}[theorem]{Lemma}
 \newtheorem{proposition}[theorem]{{\bf Proposition}}
 \newtheorem{corollary}[theorem]{Corollary}

 \newenvironment{proofof}[1]{\vspace*{5mm} \par \noindent
         {\bf Proof of #1.\hs{2}}}{\hfill$\Box$ \vspace*{3mm}}

\newcommand{\ignore}[1]{}

\newcommand{\cent}{{|}\!\!\mathrm{c}}
\newcommand{\dollar}{\$}

\newcommand{\mmid}{\!\mid\!}

\newcommand{\boldvec}[1]{\mbox{\boldmath $ #1 $}}

\begin{document}
%%%%%%%%%%%%%%%%%%
%%%%%%%%%%%%%%%%%%

\pagestyle{plain}
\pagenumbering{arabic}
\setcounter{page}{1}
\setcounter{footnote}{0}

\begin{center}
{\Large {\bf The No Endmarker Theorem for One-Way \s\\ Probabilistic Pushdown Automata}}
 \bs\ms\\

{\sc Tomoyuki Yamakami}\footnote{Present Affiliation: Faculty of Engineering, University of Fukui, 3-9-1 Bunkyo, Fukui 910-8507, Japan}
\bs\\
\end{center}

%%%%%%%%%%%%%%%%%%%%%%%%%%%%%%%%%%%%%%%%%%%%%%%
%%%%%%%%%%%%%%%%%%
\sloppy
\begin{abstract}
In various models of one-way pushdown automata, the explicit use of two designated endmarkers on a read-once input tape has proven to be    extremely useful for making a conscious, final decision on the acceptance/rejection of each input word immediately after reading the right endmarker. With no endmarkers, by contrast,
a machine must constantly stay in either accepting or rejecting states at any moment since it never notices the end of the input word. This situation, however, helps us analyze the behavior of the machine whose tape head makes the consecutive moves on all prefixes of a given extremely long input word.
Since those two machine formulations have their own advantages, it is natural to ask whether the endmarkers are truly necessary to correctly recognize languages.
In the deterministic and nondeterministic models, it is well-known that the endmarkers are removable without changing the acceptance/rejection criteria of each input word.
This paper demonstrates that, for a more general model of one-way probabilistic pushdown automata, the endmarkers are always removable. This is proven by
employing probabilistic transformations from an ``endmarker'' machine to an equivalent ``no-endmarker'' machine at the cost of double exponential stack-state complexity without compromising its error probability.
By setting the machine's error probability appropriately, our proof also provides an alternative proof to both the deterministic and the nondeterministic models as well.

\s

\n{\bf Keywords.}
probabilistic pushdown automata, endmarker, descriptional complexity, acceptance/rejection  criteria, stack-state complexity
\end{abstract}

%%%%%%%%%%%%%%%%%
%%%%%%%%%%%%%%%%%
\sloppy

\section{With or Without Endmarkers, That is a Question}\label{sec:introduction}

%%%%%
\subsection{Two Different Formulations of One-Way Pushdown Automata}\label{sec:ppda-intro}

In automata theory, \emph{pushdown automata} are regarded as one of the most fundamental architectures operated with finite-state controls.
A pushdown automaton is a simple finite-controlled machine equipped with a special memory device, called a \emph{stack}, in which  information is stored and modified in the first-in last-out manner. Here, we are focused on one-way pushdown automata whose tape heads either move to the right or stay still at any step. In particular, a nondeterministic variant of those machines, called  \emph{one-way nondeterministic pushdown automata} (or 1npda's), are known to characterize well-studied context-free languages.
Similarly, a deterministic variant, known as \emph{one-way deterministic pushdown automata} (or 1dpda's), induce deterministic context-free languages. These machines can be seen as special cases of a much more general model of
\emph{one-way probabilistic pushdown automata} (or 1ppda's). Macarie and Ogihara \cite{MO98} discussed the computational complexity of languages recognized  by these 1ppda's with unbounded-error probability.
In recent literature, Hromkovi\v{c} and Schnitger \cite{HS10} and Yamakami \cite{Yam17} further discussed the limitations of the computational power of bounded-error 1ppda's.

In many textbooks and scientific papers (ignoring small variational deviations), various one-way pushdown automata have two distinct but widely-used formulations.
At the first glance, those formulations look quite different and it is not immediately obvious that they are essentially ``equivalent'' in recognition power. In the first formalism, an input string over a fixed alphabet is initially given to a read-only input tape together with two designated \emph{endmarkers}, which mark both ends of the input string. The machine starts at scanning the left endmarker $\cent$ and moves its tape head to the right whenever it accesses an input symbol until it enters a halting state  (i.e., either an accepting state or a rejecting state). When the machine enters a halting state, the machine is considered to ``halt'' and its computation terminates in acceptance or rejection depending on the type of halting states.
The machine may be allowed to
proceed without reading any input symbol, even after reading the right endmarker $\dollar$.
A series of such special steps is referred to as a final series of   \emph{$\lambda$-moves} or a \emph{$\lambda$-transitions}, which can help the machine, for example, empty its stack before halting.
Even if the machine halts before reading the right endmarker by simply entering halting states, it is possible to postpone the halting time until the machine reads the right endmarker because of the existence of the right endmarker.

In the second formalism, by sharp contrast, the input tape contains only a given input string without the presence of two endmarkers and the machine starts reading the input from left to right until it reads off the rightmost input symbol. Since the machine needs to access the entire input string without knowing the end of the input, whenever the machine reads off the input, it is considered to ``halt'' and the current inner state of the machine determines the acceptance and rejection of the input. As in the first formulation, the machine may be either allowed or disallowed to make a series of $\lambda$-moves just after reading the rightmost input symbol.

For convenience, in this paper, we call a model described by the first formalism an   \emph{endmarker model} and one by the second formalism a \emph{no-endmarker model}. See Fig.~\ref{fig:difference-1-ppda} for the difference between these two models.
Ka\c{n}eps, Geidmanis, and Freivalds \cite{KGF97}, Hromkovi\v{c} and Schnitger \cite{HS10}, and Yamakami \cite{Yam17} all used a no-endmarker model of 1ppda's (succinctly called \emph{no-endmarker 1ppda's}) in their analyses of machine's behaviors, whereas an endmarker model of  1ppda's (called \emph{endmarker 1ppda's}) was used by Macarie and Ogihara \cite{MO98} and Yamakami \cite{Yam19}. It is commonly assumed that endmarker 1ppda's are equivalent in recognition power to no-endmarker 1ppda's. Unfortunately, there is no ``formal''  proof for the equivalence between those two models \emph{without compromising error probability}.

The roles of two endmarkers, the left and the right endmarkers, are clear. In the presence of the right endmarker $\dollar$, in particular, whenever the machine reads $\dollar$, the machine surely notices the end of the input string and, based on this knowledge, it can make the final transition (followed by a possible series of $\lambda$-moves) before entering either accepting or rejecting states.
In the presence of $\dollar$, moreover, we can make a series of $\lambda$-moves to empty the stack before halting.
Without the right endmarker, on the contrary, the machine must be always in a ``constantly halting'' inner state (either an accepting state or a rejecting state). The acceptance or rejection of the machine is determined by only  such an inner state obtained just after reading off the entire input string and then following a (possible) series of $\lambda$-moves.
This halting condition certainly helps us analyze the behavior of the machine simply by tracing down the changes of accepting and rejecting states on all prefixes of any extremely long input string. For instance, some of the results of Hromkovi\v{c} and Schnitger \cite{HS10} and Yamakami \cite{Yam17} take advantages of this characteristic, structural property.

For 1dpda's as well as 1npda's, Hopcroft and Ullman \cite{HU79} used  in their textbook the no-endmarker model to define ``pushdown automata'' whereas Lewis and Papadimitriou \cite{LP98} suggested in their textbook the use of the endmarker model.
For the basic model of pushdown automata, the endmarker model and the no-endmarker model are in fact well-known to be ``equivalent'' in their computational power.
We succinctly call this assertion the \emph{no endmarker theorem} throughout this paper.

As for 1npda's, for instance, we can easily convert one model to the other without compromising its acceptance/rejection criteria by first transforming a 1npda $M$ to its equivalent context-free grammar, converting it to Greibach Normal Form, and then translating it back to its equivalent 1npda (see, e.g., \cite{HU79}).

The question of whether the endmarker models of one-way pushdown automata can be computationally equivalent to the no endmarker models of the same machine types is so fundamental and useful, particularly, in the study of various types of one-way pushdown automata.
This paper extends our attention to 1ppda's---a probabilistic variant of pushdown automata.
For such 1ppda's, we wish to argue whether the no endmarker theorem is indeed true because, unfortunately, not all types of pushdown automata enjoy this useful theorem.
The right endmarker can be eliminated if those machine models are both closed under right quotient with regular languages (see, e.g., \cite{HU79}). Since 1dpda's and 1npda's satisfy such a closure property \cite{GG66}, the no endmarker theorem holds for these machines.
For counter machines (i.e., pushdown automata with single letter stack alphabets except for the bottom marker), on the contrary, the endmarkers are generally unremovable \cite{EIM19}; more precisely, the deterministic reversal-bounded multi-counter machines with endmarkers are generally more powerful than the same machines with no endmarkers. This latter fact demonstrates the crucial role of the endmarkers
for pushdown automata in general.
It is therefore of importance to discuss the removability of the endmarkers for bounded-error and unbounded-error 1ppda's.

%%%%%
\subsection{The No Endmarker Theorem}\label{sec:intro-no-endmarker}

This paper presents the proof of the \emph{no endmarker theorem} for 1ppda's with arbitrary error probability. More precisely, we intend to prove the following statement, which allows us to safely remove the two endmarkers of 1ppda's without compromising error probability.

\begin{theorem}\label{no-endmarker}{\rm [No Endmarker Theorem]}
Let $\Sigma$ denote any alphabet and let $\alpha,\gamma:\Sigma^*\to[0,1]$ be two arbitrary functions satisfying  $0 < \alpha(x)+\gamma(x)\leq1$ for all $x\in\Sigma^*$.
The following two statements are logically equivalent.

\begin{enumerate}\vs{-2}
  \setlength{\topsep}{-2mm}%
  \setlength{\itemsep}{0mm}%
  \setlength{\parskip}{0cm}%

\item There exists a 1ppda with two endmarkers that produces acceptance probability $\alpha(x)$ and rejection probability $\gamma(x)$ on every input $x$.

\item There exists a 1ppda with no endmarker that produces
acceptance probability $\alpha(x)$ and rejection probability $\gamma(x)$ on every input $x$.
\end{enumerate}
\end{theorem}

Theorem \ref{no-endmarker} roughly asserts that we can safely eliminate the two endmarkers from each 1ppda without changing the original acceptance and rejection probabilities. The invariance of these  probabilities is important because  this invariance makes it possible to apply the same proof of ours to both 1dpda's and 1npda's as special cases by setting their associated \emph{acceptance/rejection criteria}\footnote{A 1dpda \emph{accepts} (resp., \emph{rejects}) an input $x$ if $\alpha(x)=1$ (resp., $\gamma(x)=1$), whereas a 1npda \emph{accepts} (resp., \emph{rejects}) $x$ if $\alpha(x)>0$ (resp., $\gamma(x)=1$).}  appropriately.
As a corollary, we immediately obtain the well-known fact that the endmarkers are removable for 1dpda's and 1npda's.

\begin{corollary}\label{main-corollary}
Theorem \ref{no-endmarker} holds also for 1dpda's as well as 1npda's with their own acceptance/rejection criteria.
\end{corollary}

Since our proof of Theorem \ref{no-endmarker} is constructive, it is possible to discuss the  increase of the number of inner states and of stack alphabet size in the construction of new machines. Another important factor is the maximal size of stack strings stored by ``push'' operations. We call this number the \emph{push size} of the machine. Throughout this paper, altogether of those three factors are succinctly referred to as the \emph{stack-state  complexity} of transforming one model to the other.

For brevity, we say that two 1ppda's (with or without endmarkers) are \emph{error-equivalent} if their outputs agree with each other on all inputs with exactly the same acceptance and rejection  probabilities.

\begin{proposition}\label{adding-endmarker}
Given an $n$-state no-endmarker 1ppda with stack alphabet size $m$ and push size $e$, there is an error-equivalent endmarker 1ppda of stack alphabet size $m+2$ and push size $e$ with at most $2n+1$ states.
\end{proposition}

\begin{proposition}\label{deleting-endmarker}
Given an $n$-state endmarker 1ppda with stack alphabet size $m$ and push size $e$ over input alphabet $\Sigma$, there is an error-equivalent no-endmarker 1ppda with at most $262 d^22^{26d}$ inner states, at most  stack alphabet size $d^42^{52d}$, and push size $3$, where $d=144|\Sigma|e^2n^3m^3(2m)^{16enm}$.
\end{proposition}

Since Theorem \ref{no-endmarker} follows directly from Propositions \ref{adding-endmarker}--\ref{deleting-endmarker}, we will concentrate our efforts on proving these two propositions in the rest of this paper.
To simplify the complexity description of a pushdown automaton, similarly to the \emph{state complexity} of a finite automaton, we introduce the notion of the \emph{stack-state complexity} of a pushdown automaton equipped with a set $Q$ of inner states, a set $\Gamma$ of stack symbols, and push size $e$ by taking the value $|Q||\Gamma^{\leq e}|$, where $\Gamma^{\leq e}=\{w\in\Gamma^*\mid |w|\leq e\}$.
From Proposition \ref{deleting-endmarker}, we observe that the stack-state complexity of transforming an endmarker 1ppda to another error-equivalent no-endmarker 1ppda is $O(d^{18}2^{234d})$.
This suggests that endmarker 1ppda's are likely to be ``double-exponentially'' more succinct in descriptional complexity than no-endmarker 1ppda's. We are not sure, however, that this bound can be significantly improved.
For more concrete, it is still unknown that we can reduce this double-exponential factor to an exponential factor or even a polynomial factor. See Section \ref{sec:open-problem} also for this issue.

\paragraph{Organization of This Paper}
In Section \ref{sec:1ppda}, we will start with the formal definition of 1ppda's (and their variants, 1dpda's and 1npda's) with or without endmarkers. In Section \ref{sec:adding-endmarker}, we will prove Proposition \ref{adding-endmarker}.  Proposition \ref{deleting-endmarker} will be proven in Section \ref{sec:final-proof}. To simplify the proof of this proposition, we will transform each standard 1ppda into another special 1ppda that does not halt before the right endmarker and also  takes a  ``push-pop-controlled'' form, called an \emph{ideal shape}, which turns out to be quite useful in proving various properties of languages.
As a concrete example of usefulness, we will demonstrate in Section \ref{sec:usefulness} the closure property of language families induced by unbounded-error endmarker 1ppda's under ``reversal''. Lastly, in Section \ref{sec:open-problem}, a few open questions regarding the subjects of this paper will be discussed.

%%%%%%%%%%%%%%%%%
%%%%%%%%%%%%%%%%%
\section{Two Formulations of One-Way Pushdown Automata}\label{sec:1ppda}

Let us review two machine models of 1ppda's in detail, focusing on the use of the two endmarkers.
A \emph{one-way probabilistic pushdown automaton} (abbreviated as 1ppda) $M$ runs essentially in a way similar to a one-way nondeterministic pushdown automaton (or a 1npda)  except that, instead of making a nondeterministic choice  at each step of forming a number of computation paths, $M$ randomly chooses one of all possible transitions and then branches out to produce multiple computation paths. The probability of such a computation path is determined by a series of  random choices made along this computation path.
The past literature has taken two different formulations of 1ppda's,  \emph{with} or \emph{without} two endmarkers. We will explain these two formulations in the subsequent subsections.

%%%%%
\subsection{Numbers, Strings, and Languages}\label{sec:numbers}

Let $\nat$ denote the set of all \emph{natural numbers}, including $0$, and set $\nat^{+}$ to be $\nat-\{0\}$.
Given two integers $m$ and $n$ with $m\leq n$, the notation $[m,n]_{\integer}$ expresses the \emph{integer interval} $\{m,m+1,m+2,\ldots,n\}$.
When $n\geq1$, in particular, we further abbreviate $[1,n]_{\integer}$ as $[n]$.
We also use the real intervals, such as $[0,1]$, $[0,1/2)$, etc.
Given a set $A$, $\PP(A)$ denotes the \emph{power set} of $A$; namely, the set of all subsets of $A$.

An \emph{alphabet} is a finite nonempty set of ``letters'' or ``symbols.'' Given an alphabet $\Sigma$, a \emph{string} over $\Sigma$ is a finite sequence of symbols in $\Sigma$. Conventionally, we use $\lambda$ to express the \emph{empty string} as well as the special pop operand used for a stack. The \emph{length} of a string $x$, expressed as $|x|$, is the total number of the occurrences of any symbols in $x$. In particular, $|\lambda|=0$ follows. 
Given a nonempty string $x$ and an integer $i\in[|x|]$, $x_{(i)}$ denotes the $i$th symbol of $x$. For convenience, we also set $x_{(0)}=\lambda$. It follows that $x$ coincides with $x_{(1)}x_{(2)}\cdots x_{(|x|)}$.
For a string $x=x_{(1)}x_{(2)}\cdots x_{(n-1)}x_{(n)}$ over $\Sigma$ with $n=|x|$, the \emph{reverse} of $x$ is the string $x_{(n)}x_{(n-1)}\cdots x_{(2)}x_{(1)}$ and is denoted by $x^R$. In particular, $\lambda^R$ is the same as $\lambda$.

The notation $\Sigma^*$ indicates the set of all strings over $\Sigma$ whereas $\Sigma^{+}$ expresses $\Sigma^*-\{\lambda\}$. A \emph{language} over $\Sigma$ is just a subset of $\Sigma^*$.
For two languages $A$ and $B$, $AB$ denotes the \emph{concatenation} of $A$ and $B$, that is, $\{xy\mid x\in A,y\in B\}$. In particular, when $A=\{w\}$, we write $wB$ instead of $\{w\}B$. Similarly, when $B=\{w\}$, we use the succinct notation $Aw$.
For a number $k\in\nat$, $A^k$ expresses the set $\{x_1x_2\cdots x_k \mid x_1,x_2,\ldots,x_k\in A\}$.
Moreover, we use the notation $[A,B]$ for two sets $A$ and $B$ to denote the set $\{[a,b]\mid a\in A,b\in B\}$ of ``bracketed'' ordered pairs. Similarly, we use $[A,B,C]$ for the set $\{[a,b,c]\mid a\in A,b\in B,c\in C\}$.

%%%%
%%%%
\subsection{The First Formulation}\label{sec:first-formulation}

We start with explaining the first formulation of 1ppda's whose input tapes are marked by two designated endmarkers ($\cent$ and $\dollar$). We always assume that these endmarkers are not included in any input alphabet $\Sigma$. As discussed in, e.g., \cite{LP98}, a 1ppda with two endmarkers (which we call an \emph{endmarker 1ppda}) has a read-once semi-infinite input tape, on which an input string is initially placed, surrounded by two endmarkers
$\cent$ (left endmarker) and $\dollar$ (right endmarker). Fig.~\ref{fig:difference-1-ppda}(1) illustrates an endmarker 1ppda.
The tape region is called an \emph{input region} if it contains $\cent x\dollar$ including an input string $x$.
To clarify the use of those two endmarkers, we explicitly include them in the description of a machine as
$(Q,\Sigma,\{\cent,\dollar\},\Gamma,\Theta_{\Gamma}, \delta,q_0,Z_0,Q_{acc},Q_{rej})$, where $Q$ is a finite set of inner  states, $\Sigma$ is an input alphabet, $\Gamma$ is a stack alphabet, $\Theta_{\Gamma}$ is a finite subset of $\Gamma^*$ with $\lambda\in \Theta_{\Gamma}$, $\delta$ is a \emph{probabilistic transition function} from $(Q-Q_{halt})\times \check{\Sigma}_{\lambda} \times\Gamma \times Q\times  \Theta_{\Gamma}$ to $[0,1]$ (with $\check{\Sigma}=\Sigma\cup\{\cent,\dollar\}$ and $\check{\Sigma}_{\lambda} =\check{\Sigma}\cup\{\lambda\}$), $q_0$ ($\in Q$) is the \emph{initial state}, $Z_0$ ($\in\Gamma$) is the \emph{bottom marker}, $Q_{acc}$ ($\subseteq Q$) is a set of accepting states, and $Q_{rej}$ ($\subseteq Q$) is a set of rejecting states,
where $Q_{halt}$ denotes $Q_{acc}\cup Q_{rej}$.
Any element of $Q_{halt}$ is called a \emph{halting state}.
As customary, we always assume that $Q_{acc}\cap Q_{rej}=\setempty$. For convenience, let $\Gamma^{(-)}=\Gamma-\{Z_0\}$.
Each value $\delta(q,\sigma,a, p,u)$ expresses the probability that, when $M$ reads $\sigma$ on the input tape and $a$ in the top of the stack in inner state $q$, $M$ changes $q$ to another inner state $p$, and replaces $a$ by a string $u$.
For the clarity reason, we express $\delta(q,\sigma,a,p,u)$ as $\delta(q,\sigma,a \mmid p,u)$, because this emphasizes the circumstance that $(q,a)$ is changed into $(p,u)$ after scanning $\sigma$.
The machine $\delta$, $M$ is considered to ``halt'' when it enters a halting state. This is the reason why $\delta$ is not defined on $Q_{halt}$. The minimum positive integer $e$ for which $\Theta_{\Gamma}\subseteq \Gamma^{\leq e}$ is called the \emph{push size} of $M$.

Concerning the ``first'' step of $M$, in this paper, we always force $M$ to ``read'' the left endmarker $\cent$. In other words,  no $\lambda$-move can be made at the first step.

We demand that $Z_0$ cannot be removed from the stack or replaced by any other symbols; that is, whenever $\delta(q,\sigma,a \mmid p,u)>0$,   $u$ must belong to $(\Gamma^{(-)})^*$ for any $a\neq Z_0$, and $u$ must be in $(\Gamma^{(-)})^*Z_0$ for $a=Z_0$. If $\sigma=\lambda$, then the tape head must stay still and we call this  transition a \emph{$\lambda$-move} (or a \emph{$\lambda$-transition}); otherwise, the tape head must move to an adjacent cell. Notice that, unless making any $\lambda$-move, $M$ always moves its tape head in one direction, from left to right.

At each step, $M$ probabilistically selects $\lambda$-moves as well as non-$\lambda$-moves with certain fixed probabilities.
Conveniently, we set $\delta[q,\sigma,a]$ to be $\sum_{(p,u)\in Q\times \Theta_{\Gamma}} \delta(q,\sigma,a \mmid p,u)$ for any triplet $(q,\sigma,a)\in (Q-Q_{halt}) \times \check{\Sigma}\times\Gamma$. When $\sigma\neq \lambda$, $\delta[q,\sigma,a]$ indicates the probability of making non-$\lambda$-moves; by contrast, $\delta[q,\lambda,a]$ indicates that of making $\lambda$-moves.
We always demand that $\delta$ satisfies the following \emph{probability requirement}:  $\delta[q,\sigma,a]+\delta[q,\lambda,a]=1$ for any $(q,\sigma,a)\in (Q-Q_{halt}) \times \check{\Sigma} \times \Gamma$.

In the literature, after reading $\dollar$, pushdown automata are either allowed or disallowed to make any series of $\lambda$-moves until they halt. In this paper, such a series of $\lambda$-moves is distinctively called a \emph{final series of $\lambda$-moves} and is treated separately.
For a general treatment of 1ppda's, unless otherwise stated, this paper ``allows'' the machine to  make such a final series of $\lambda$-moves (after reading $\dollar$) until it finally enters a halting state.
After reading $\dollar$, the tape head is considered to \emph{move off} (or \emph{leave}) an input region.

To describe the behaviors of the stack, we follow \cite{Yam16} for the basic terminology. A \emph{stack content} means a series $z=z_mz_{m-1}\cdots z_1z_0$ of stack symbols in $\Gamma$, which are stored inside the stack sequentially from the topmost symbol $z_m$ of the stack to the lowest symbol $z_0$ ($=Z_0$).
Since the bottom marker $Z_0$ cannot be popped, we say that the stack is \emph{empty} exactly when there is no symbol in the stack except for $Z_0$.

All tape cells on an input tape are indexed sequentially by natural numbers from left to right in such a way that $\cent$ is located in the $0$th tape cell (called the \emph{start cell}), a given input $x$ is placed from cell $1$ to cell $|x|$, and $\dollar$ is in the $|x|+1$st cell.
A \emph{(surface) configuration} of $M$ is a triplet $(q,i,w)$, which indicates that $M$ is in inner state $q$, $M$'s tape head scans the $i$th tape cell, and $w$ is $M$'s current stack content. The \emph{initial configuration} of $M$ is therefore $(q_0,0,Z_0)$. An \emph{accepting} (resp., a \emph{rejecting}) \emph{configuration} is a configuration with an accepting (resp., a rejecting) state and a \emph{halting configuration} is either an accepting or a rejecting configuration.
Given a fixed input $x\in\Sigma^*$, we say that a configuration $(p,j,uw)$ \emph{follows} another configuration $(q,i,aw)$ with probability $\delta(q,\sigma,a \mmid p,u)$ if $\sigma$ is the $i$th symbol of $x$, and $j=i$ if $\sigma=\lambda$ and $j=i+1$ if $\sigma\neq\lambda$.
For two configurations $(q,i_0,w)$ and $(p,j,u)$ of $M$, we write $(q,i,w)\vdash^{\varepsilon}_{M} (p,j,u)$ if $(p,j,u)$ follows $(q,i,w)$ by a single transition of $M$ with probability $\varepsilon$. When $\varepsilon>0$, we instead write $(q,i,w)\vdash^{>0}_{M} (p,j,u)$.

A \emph{computation path} of length $k$, which describes a history of consecutive ``moves'' chosen by $M$ on input $x$,  is formally a series of $k+1$ configurations, starting at the initial configuration with probability $p_0=1$  and, for each index $i\in[0,k-1]_{\integer}$,  the $i+1$st configuration follows the $i$th configuration with probability $p_{i+1}$, ending at a halting configuration with probability $p_{k}$, namely,
a series $(q_0,i_0,w_0), (q_1,i_1,w_1),(q_2,i_2,w_2),\ldots, (q_k,i_k,w_k)$ of configurations satisfying $(q_j,i_j,w_j)\vdash^{p_j}_{M} (q_{i+1},i_{j+1},w_{j+1})$ for any index $j\in[0,k-1]_{\integer}$.
To such a computation path, we assign the probability $\Pi_{i\in[k]}p_i$, which is called the \emph{probability of the computation path}.
A computation path is called \emph{accepting} (resp., \emph{rejecting}) if the path ends with an accepting configuration (resp., a rejecting configuration).
Generally, a 1ppda may produce an extremely long computation path or even an infinite computation path; therefore, we must restrict our attention to finite computation paths.

Hromkovi\v{c} and Schnitger \cite{HS10} and Ka\c{n}eps, Geidmanis, and Freivalds \cite{KGF97} all used a model of 1ppda's whose \emph{computation paths all halt eventually (i.e., in finitely many steps) on every input}. We also adopt this convention\footnote{Even if we allow infinitely long computation paths with a relative small probability, we can still simulate two models as stated in Theorem \ref{no-endmarker}, and thus the theorem still holds.} in the rest of this paper.
In what follows, we always assume that all 1ppda's should satisfy this requirement. Standard definitions of 1dpda's and 1npda's do not have such a runtime bound, because we can easily convert those machines to ones that halt within $O(n)$ time (e.g., \cite{HU79,Yam08,Yam16}).

Under this requirement, when we handle 1npda's $\lambda$-moves, for any substring (or segment) $x$ of input $\cent w\dollar$, we say that \emph{$x$ is completely read} (or \emph{read off}) if  $M$ reads all symbols in $x$, moves its tape head off the string $x$, and makes all (possible) finite series of $\lambda$-moves just after reading the last symbol of $x$.

The \emph{acceptance probability} of $M$ on input $x$ is the sum of all probabilities of accepting computation paths of $M$, starting with $\cent{x}\dollar$ on its input tape.  We express by $p_{M,acc}(x)$ the acceptance probability of $M$ on $x$. Similarly, we define $p_{M,rej}(x)$ to be the \emph{rejection probability} of $M$ on $x$.
Whenever $M$ is clear from the context, we often omit script ``$M$'' entirely and write  $p_{acc}(x)$ and $p_{rej}(x)$ instead of $p_{M,acc}(x)$ and $p_{M,rej}(x)$, respectively.

To express the behavioral equivalence of probabilistic machines, we use the following terminology.
Two 1ppda's $M_1$ and $M_2$ are said to be \emph{error-equivalent} to each other if $p_{M_1,\tau}(x)=p_{M_2,\tau}(x)$ holds for all $x$ and $\tau\in\{acc,rej\}$.
In this paper, we particularly concern ourselves with the descriptional complexity of machine models.
We use the following three complexity measures to describe each machine $M$. The \emph{state size} of $M$ is $|Q|$, the \emph{stack alphabet size} of $M$ is $|\Gamma|$, and the \emph{push size} of $M$ is the maximum length of any string in $\Theta_{\Gamma}$, provided that all elements of $\Theta_{\Gamma}$ appear in the values of $\delta$.
In other words, the push size is the maximum size of any strings pushed into the stack in single moves.  It then follows that $|\Theta_{\Gamma}|\leq m^e$ if $M$ has stack alphabet size $m$ and push size $e$.
To combine them all, it is worth considering the \emph{stack-state complexity} of $M$, which is defined to be $|Q||\Gamma^{\leq e}|$. This notion is a natural extension of the state complexity of a finite automaton.

%%%%%%
%%%%%%

\begin{figure}[t]
\centering
\includegraphics*[height=3.0cm]{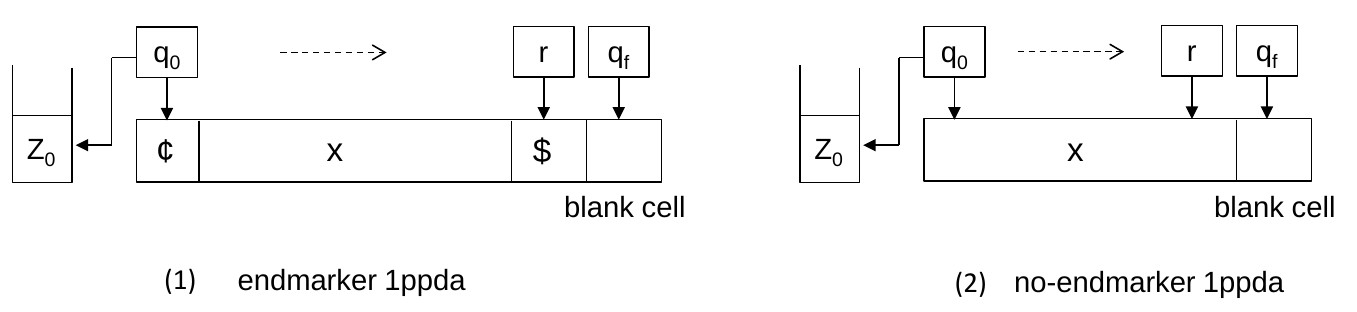}
\caption{(1) An endmarker 1ppda starts its computation by reading $\cent$. For each finite computation path, if the 1ppda does not halt before reading $\dollar$, then it reads $\dollar$, makes a (possible) final series of $\lambda$-moves, and enters a halting state. In this case, however, no blank cell outside of the scope between $\cent$ and $\dollar$ should be read. (2)  A no-endmarker 1ppda starts reading the leftmost symbol of $x$ and continues reading all symbols of $x$. The final acceptance/rejection of the 1ppda is determined after a (possible) final series of $\lambda$-moves, however, no blank cell outside of the input $x$ should be read.}\label{fig:difference-1-ppda}
\end{figure}

%%%%%
%%%%%

%%%
\subsection{The Second Formulation}\label{sec:second-formulation}

In comparison with the first formulation, let us consider 1ppda's with no endmarker. Such machines are succinctly called \emph{no-endmarker 1ppda's} in this paper and they are naturally obtained from all the definitions stated in Section \ref{sec:first-formulation} except for removing the entire use of $\cent$ and $\dollar$.
Fig.~\ref{fig:difference-1-ppda}(2) illustrates such a no-endmarker 1ppda.
We express such a machine $M$ as $(Q,\Sigma,\Theta_{\Gamma},\delta,q_0,Z_0,Q_{acc},Q_{rej})$ without $\cent$ and $\dollar$. Notice that an \emph{input region} is now the cells that contain $x$ instead of $\cent x \dollar$.
For a no-endmarker 1ppda, its probabilistic transition function $\delta$ maps $Q\times \Sigma_{\lambda}\times \Gamma\times Q\times\Theta_{\Gamma}$ to $[0,1]$, where $\Sigma_{\lambda}=\Sigma\cup\{\lambda\}$.
As for the indexing of all tape cells, since there is no endmarker, we do not use ``cell $0$''. A nonempty string $x$ is written on the input tape from cell $1$ to cell $|x|$.
Similar to $\delta$ of the first endmarker 1ppda in Section \ref{sec:first-formulation}, $\delta(q,\sigma,a,p,u)$ (or $\delta(q,\sigma,a\mmid p,u)$) indicates the probability of $M$ changing a set of its current parameters $(q,\sigma x,aw)$ to another $(p,x,uw)$ for any strings $x$ and $w$. We use the notation $\delta[q,\sigma,a]$ for the value $\sum_{(p,u)\in Q\times \Theta_{\Gamma}} \delta(q,\sigma,a\mmid p,u)$.  The probability requirement for a no-endmarker 1ppda is $\delta[q,\sigma,a]+\delta[q,\lambda,a]=1$ for any $(q,\sigma,a)$ in $Q\times\Sigma\times\Gamma$.
Similarly to the endmarker model, in general, a no-endmarker 1ppda either allows or disallows any \emph{final series of $\lambda$-moves}  (i.e., any series of $\lambda$-moves after completely reading an input).
We always assume that all computation paths of $M$ terminate in finitely many steps (similarly to the first endmarker model).
The sum of all accepting (resp., rejecting) computation paths on input $x$ forms the acceptance (resp., rejection) probability $p_{M,acc}(x)$ (resp., $p_{M,rej}(x)$).
It is important to remark that the \emph{acceptance} and \emph{rejection} of each computation path of $M$ are determined by whether $M$ respectively enters  accepting states and rejecting states just after reading off the entire input string, and leaving the input region, and making a (possible) series of subsequent $\lambda$-moves.
To ensure this condition, $Q$ should be partitioned into $Q_{acc}$ and $Q_{rej}$; namely, $Q=Q_{acc}\cup Q_{rej}$ and $Q_{acc}\cap Q_{rej}=\setempty$. This means $Q_{halt}=Q$.
Hence, $\delta$ must be defined on all elements of $Q$. The \emph{push size} and the \emph{stack-state complexity} of $M$ are defined in the same way as the first formulation.

We need to pay a special attention to the behavior of the very first step of $M$. In this paper, we demand that every 1ppda in the initial state $q_0$ should ``read''  the leftmost symbol of a given nonempty input string written on the input tape before taking any $\lambda$-moves.
Therefore, whenever an input is the empty string $\lambda$,  any machine that has no left endmarker but has the right endmarker  can still read $\dollar$ at the first step; however,  if the machine has no endmarker, then we need to allow the machine to make a $\lambda$-move at the first step because no machine is capable of  ``reading'' a blank symbol.

%%%%%
\subsection{Acceptance/Rejection Criteria and Deterministic and Nondeterministic Variants}

Throughout Sections \ref{sec:first-formulation}--\ref{sec:second-formulation}, we have already described the two different formulations of 1ppda's. In this subsection, we discuss issues applicable to the both formulations.
For a 1ppda $M$, we ``generally'' say that $M$ \emph{accepts} (resp., \emph{rejects}) $x$ if $p_{M,acc}(x)> p_{M,rej}(x)$ (resp., $p_{M,acc}(x)\leq p_{M,rej}(x)$).
Since  all computation paths are always assumed to  halt in finitely many steps, for any given string $x$, either $M$ accepts it or $M$ rejects it. The notation $L(M)$ stands for the set of all strings $x$ accepted by $M$; that is,  $L(M)=\{x\in\Sigma^* \mid \text{ $M$ accepts $x$ }\}$. Given a language $L$, $M$ \emph{recognizes} $L$ if $L$ coincides with the set $L(M)$.
The 1ppda $M$ is said to make \emph{unbounded error} if  either $p_{M,acc}(x)>1/2$ or $p_{M,rej}(x)\geq 1/2$ for all inputs $x$. In contrast, we say that $M$ makes \emph{bounded error} if there exists a real constant $\varepsilon\in[0,1/2)$ (called an \emph{error bound}) such that, for every input $x$, either $p_{M,acc}(x)\geq 1-\varepsilon$ or $p_{M,rej}(x)\geq 1-\varepsilon$.

Deterministic and nondeterministic variants of 1ppda's can be easily obtained by appropriately  modifying the acceptance/rejection criteria of 1ppda's.
To obtain a \emph{one-way deterministic pushdown automaton} (or a 1dpda), we additionally require $\delta(q,\sigma,a \mmid p,u)\in\{0,1\}$ (instead of the unit real interval $[0,1]$) for all tuples $(q,\sigma,a,p,u)$. Therefore, it immediately follows that a 1dpda $M$  accepts input $x$ if  $p_{M,acc}(x)=1$ and rejects $x$ if $p_{M,acc}(x)=0$ (implying $p_{M,rej}(x)=1$).
Similarly, we can obtain \emph{one-way nondeterministic pushdown automata} (or 1npda's) if
we pick a possible next move uniformly at random and
take the following ``nondeterministic'' acceptance/rejection criteria: a 1npda $M$ accepts input $x$ if $p_{M,acc}(x)>0$ and $M$ rejects $x$ if $p_{M,acc}(x)=0$ (and thus $p_{M,rej}(x)=1$). It is not difficult to see that these ``probabilistic'' definitions of 1dpda's and 1npda's coincide with the ``standard'' definitions presented in many textbooks. Because of these definitions, Corollary \ref{main-corollary} comes directly from Theorem \ref{no-endmarker}.

%%%%%%
%%%%%%
\section{From No-Endmarker 1ppda's to Endmarker 1ppda's}\label{sec:adding-endmarker}

Proposition \ref{adding-endmarker} ensures a transformation of no-endmarker 1ppda's into error-equivalent endmarker 1ppda's and, more importantly, asserts a linear increase of the stack-state complexity of this transformation.
Although the transformation seems rather easy,  to calculate an estimation of the desired stack-state complexity, we wish to describe
how to construct from any no-endmarker 1ppda $M$ an error-equivalent endmarker 1ppda $N$.

\begin{proofof}{Proposition \ref{adding-endmarker}}
Given any no-endmarker 1ppda  $M$ of the form $(Q,\Sigma,\Gamma,\Theta_{\Gamma}, \delta,q_0, Z_0,Q_{acc},Q_{rej})$, we wish to construct an equivalent endmarker 1ppda $N$ of the form $(Q',\Sigma,\{\cent,\dollar\},\Gamma',\Theta_{\Gamma'}, \delta', q'_0,Z_0, Q'_{acc},Q'_{rej})$, where $Q'_{acc}=\{\hat{q}\mid q\in Q_{acc}\}$,  $Q'_{rej}=\{\hat{q}\mid q\in Q_{rej}\}$, $Q' = Q\cup \{q'_0\}\cup Q'_{acc}\cup Q'_{rej}$, and $\Gamma' = \Gamma\cup \{\cent,\dollar\}$. Notice that $|Q'|\leq 2n+1$ and $|\Gamma'|=m+2$.
From $\delta$, we formally define the probabilistic transition function $\delta'$ as follows. Let $\Sigma_{\lambda} = \Sigma\cup\{\lambda\}$. We use $\hat{q}$ as a new halting state associated with each $q\in Q_{halt}$.

\begin{enumerate}\vs{-1}
  \setlength{\topsep}{-2mm}%
  \setlength{\itemsep}{0mm}%
  \setlength{\parskip}{0cm}%

\item $\delta'(q'_0,\cent,Z_0 \mmid q_0,Z_0) =1$

\item $\delta'(q,\sigma,a \mmid p,w) = \delta(q,\sigma,a \mmid p,w)$ for $q\in Q$,  $a\in\Gamma$, and $\sigma\in\Sigma_{\lambda}$.

\item $\delta'(q,\dollar,a \mmid \hat{q},a)=1$ for $q\in Q_{halt}$ and $a\in\Gamma$ if $\delta[q,\lambda,a]=0$.
\end{enumerate}\vs{-1}

Intuitively speaking, the behavior of $N$ is described as follows.
Assume that $x$ is a nonempty input.
Let us consider any halting computation path  of $M$ having the form: $(q_0,i_0,w_0), (q_1,i_1,w_1),\ldots, (q_m,i_m,w_m)$, where $m\geq1$, $i_0=1$, $w_0=Z_0$, $i_m=|x|$, and $q_m\in Q_{halt}$.
We remark that $M$ is considered to ``halt'' after reading off $x$ and  making a (possible) final  series of subsequent $\lambda$-moves.

In the first move, $N$ uses the new initial state $q'_0$ to read the left endmarker $\cent$ and then deterministically enters $M$'s initial state $q_0$ (Line 1). This step adds an extra transition of the form  $(q'_0,i'_0,w'_0)\vdash_{M}^{>0} (q_0,i_0,w_0)$  with $i'_0=0$ and  $w'_0=Z_0$. After the first step, $N$ continues simulating each step of $M$ (Line 2).

Assume that $M$ reads the rightmost input symbol $x_{|x|}$, makes a (possible) final series of $\lambda$-moves, and then halts.
By following the corresponding moves of $M$, $N$ reads $x_{|x|}$, moves to the adjacent cell containing $\dollar$, and makes the corresponding series of $\lambda$-moves (Line 2).
By scanning $\dollar$,
$N$ deterministically enters a new halting state without changing the stack content (in Line 3). These additional steps introduce $(q_m,i_m,w_m)\vdash^{>0}_{N} (q'_{m+1},i'_{m+1},w'_{m+1})$, where $q'_{m+1}=\hat{q}_m$, $i'_{m+1}=|x\dollar|$, and $w'_{m+1}=w_m$. Thus, the given halting computation path of $M$ is converted into the corresponding halting computation path of $N$ of the form: $(q'_0,i'_0,w'_0), (q_0,i_0,w_0), (q_1,i_1,w_1),\ldots, (q_m,i_m,w_m), (q'_{m+1},i'_{m+1},w'_{m+1})$. Note that this computation path has the same probability as $M$'s. Therefore, with the same error probability, $N$ can simulate $M$ correctly.  By the definition of $\delta'$, the push size of $N$ is the same as that of $M$. It is important to note that $N$ makes no $\lambda$-moves after reading $\dollar$. This completes the proof of Proposition \ref{adding-endmarker}.
\end{proofof}

%%%%%%%%%%%%%%%%%
%%%%%%%%%%%%%%%%%
\section{Preparatory Machine Modifications}\label{sec:modifications}

The rest of this paper is dedicated to proving Proposition \ref{deleting-endmarker}. A key to our proof of this proposition is the step to ``normalize'' the wild behaviors of endmarker 1ppda's without compromising their error probability.
This preparatory step is quite crucial and it will help us prove the proposition significantly more easily in Section \ref{sec:deleting-endmarker}.
In the nondeterministic model, it is always possible to eliminate all $\lambda$-moves of 1npda's and restrict the set $Q$ of inner states on  $Q_{simple}=\{q_0,q,q_{acc}\}$ (this is a byproduct of translating a context-free grammar in Greibach Normal Form to a corresponding 1npda). For 1ppda's, however, we can neither eliminate $\lambda$-moves nor restrict $Q$ on $Q_{simple}$.
Moreover, we cannot control the number of consecutive $\lambda$-moves. Despite all these difficulties, nevertheless, we can still curtail certain behaviors of 1ppda's to control the execution of push and pop operations during their computations. Such a push-pop-controlled form is succinctly called ``ideal shape.''

For 1npda's and 1dpda's, there are a few precursors in this direction: Hopcroft and Ullman \cite[Chapter 10]{HU79} and Pighizzini and Pisoni \cite[Section 5]{PP15}. We intend to utilize the basic ideas of them and further expand these ideas to prove that all 1ppda's can be transformed into their ideal shape form.

%%%
\subsection{No Halting Before Reading \$}

We first modify a given endmarker 1ppda to another one with a certain nice property without compromising its error probability. In particular, we want to eliminate the possibility of premature halting before reading $\dollar$; namely,  to force the machine to halt only on or after reading $\dollar$ with a (possible) series of subsequent $\lambda$-moves. Remember that a stack is said to be  \emph{empty} exactly when it contains only the bottom marker.

\begin{lemma}\label{postpone-halting}
Given any endmarker 1ppda $M$ with $n$ states, stack alphabet size $m$, and  push size $e$, there exists an error-equivalent endmarker 1ppda $N$ having $2(n+1)$ states with stack alphabet size $m$ and push size $e$ such that it halts only on or after reading $\dollar$; moreover,
when $N$ halts, its stack becomes empty.
\end{lemma}

An underlying proof idea of Lemma \ref{postpone-halting} is to simulate $M$ step by step until $M$ either enters a halting state or reads the right endmarker $\dollar$. When $M$ enters a halting state before reading $\dollar$, $N$ remembers this inner state, empties the stack, continues reading the rest of input symbols, and finally enters a ``true'' halting state. On the contrary, when $M$ reaches $\dollar$, $N$ remembers the passing of $\dollar$, continues the simulation of $M$'s $\lambda$-moves. Once $M$ enters a halting state without the empty stack, $N$ remembers this inner state, empties the stack by making additional $\lambda$-moves, and enters a ``true''  halting state.

For the clarity of the state complexity term of $2(n+1)$ in Lemma \ref{postpone-halting} and also for the later use of the lemma's proof in Section \ref{sec:ideal shape-lemma}, we include the detailed proof of the lemma.

\begin{proofof}{Lemma \ref{postpone-halting}}
Let $M$ denote any endmarker 1ppda $(Q,\Sigma,\{\cent,\dollar\},\Gamma,\Theta_{\Gamma}, \delta,q_0, Q_{acc},Q_{rej})$ with $|Q|=n$, $|\Gamma| =m$, and push size $e$. In the following, we define the desired endmarker 1ppda $N$ of the form $(Q',\Sigma,\{\cent,\dollar\},\Gamma', \Theta_{\Gamma'}, \delta',q'_0, Q'_{acc},Q'_{rej})$, which can simulate $M$ with the same error probability.
Firstly, we prepare a new non-halting state $\bar{q}$ for each inner state $q$ in $Q_{halt}$, and two extra accepting and rejecting states $q'_{acc}$ and $q'_{rej}$. We then set $Q^{(0)} = \{\bar{q}\mid q\in Q_{rej}\}$ and $Q^{(1)} = \{\bar{q}\mid q\in Q_{acc}\}$. We also prepare $\bar{q}^{(\dollar)}$ for each inner state $q$ in $Q$ and define $Q^{(\dollar)} =\{\bar{q}^{(\dollar)}\mid q\in Q\}$.
Finally, we define $Q'_{acc} =  \{q'_{acc}\}$, $Q'_{rej} = \{q'_{rej}\}$,  $Q'= (Q-Q_{halt})\cup Q^{(\dollar)} \cup Q^{(0)}\cup Q^{(1)} \cup Q'_{acc}\cup Q'_{rej}$, and $\Gamma'=\Gamma$.
Recall from Section \ref{sec:first-formulation} the notation $\check{\Sigma} = \Sigma\cup\{\cent,\dollar\}$ and $\check{\Sigma}_{\lambda} = \check{\Sigma}\cup\{\lambda\}$.
In what follows, we describe how to define $\delta'$, which is intended to  simulate $M$ step by step on a given input.
Let us consider two cases, depending on whether or not $M$ is scanning $\dollar$.

\s
(1) Assume that $M$ does not yet read $\dollar$. As long as $M$ stays in non-halting states, $N$ precisely simulates $M$'s moves (Line 1). When $M$ enters a halting state $q$ from a non-halting state $p$ before reading $\dollar$, $N$ first enters its  associated inner state $\bar{q}$ (Line 2), deterministically empties the stack (Line 3), and continues reading the rest of the input string (Line 4).
When $N$ finally reaches $\dollar$ in inner state $\bar{q}$, $N$ changes it to a designated state $\hat{q}$ and then halts (Line 5), where $\hat{q}=q'_{acc}$ if $q\in Q_{acc}$, and $\hat{q}=q'_{rej}$ if $q\in Q_{rej}$.
The transitions of $\delta'$ are formally defined as follows. Let $q\in Q$, $p\in Q'$, $a\in\Gamma$, and $w\in\Theta_{\Gamma}$.

\begin{enumerate}\vs{-1}
  \setlength{\topsep}{-2mm}%
  \setlength{\itemsep}{0mm}%
  \setlength{\parskip}{0cm}%

\item[1.] $\delta'(p,\sigma,a \mmid q,w) = \delta(p,\sigma,a \mmid q,w)$ if $q\notin Q_{halt}$ and $\sigma\in\check{\Sigma}_{\lambda}$.

\item[2.] $\delta'(p,\sigma,a \mmid \bar{q},w)= \delta(p,\sigma,a \mmid q,w)$ if $\sigma\in\Sigma_{\lambda}\cup\{\cent\}$ and $q\in Q_{halt}$.

\item[3.] $\delta'(\bar{q},\lambda,a \mmid \bar{q},\lambda)=1$ if  $a\in\Gamma^{(-)}$ and $q\in Q_{halt}$.

\item[4.] $\delta'(\bar{q},\tau,Z_0 \mmid \bar{q},Z_0)=1$ if  $\tau\in \Sigma$ and $q\in Q_{halt}$.

\item[5.] $\delta'(\bar{q},\dollar,Z_0 \mmid \hat{q},Z_0)=1$, where $\hat{q}$ is defined above.
\end{enumerate}\vs{-1}

(2) On the contrary, when $M$ reads $\dollar$ in a non-halting state $p$ and then enters another inner state $q$, the new machine $N$ enters $\bar{q}^{(\dollar)}$ from $p$ (Line 6) and then simulates a (possible) series of $M$'s $\lambda$-moves using their corresponding states in $Q^{(\dollar)}$ (Line 7). When $M$ finally enters a halting state, say, $q$, $N$ also enters its corresponding state $\bar{q}^{(\dollar)}$, deterministically empties the stack (Line 8) in the same inner state $\bar{q}^{(\dollar)}$, and finally enters the halting state $\hat{q}$  (Line 9),  where $\hat{q}=q'_{acc}$ if $q\in Q_{acc}$, and $\hat{q}=q'_{rej}$ otherwise.
Formally, $\delta'$ is defined as follows.
Let $p,q,r\in Q$, $a\in\Gamma$, and $w\in\Theta_{\Gamma}$.

\begin{enumerate}\vs{-1}
  \setlength{\topsep}{-2mm}%
  \setlength{\itemsep}{0mm}%
  \setlength{\parskip}{0cm}%

\item[6.] $\delta'(p,\dollar,a \mmid \bar{q}^{(\dollar)},w) = \delta(p,\dollar,a \mmid q,w)$ if $p\notin Q_{halt}$.

\item[7.] $\delta'(\bar{q}^{(\dollar)},\lambda,a \mmid \bar{r}^{(\dollar)},w)= \delta(q,\lambda,a \mmid r,w)$ if $q\notin Q_{halt}$.

\item[8.] $\delta'(\bar{q}^{(\dollar)},\lambda,a \mmid \bar{q}^{(\dollar)},\lambda)=1$ if  $a\in\Gamma^{(-)}$ and $q\in Q_{halt}$.

\item[9.] $\delta'(\bar{q}^{(\dollar)},\lambda,Z_0 \mmid \hat{q},Z_0)=1$ if $q\in Q_{halt}$, where $\hat{q}$ is defined above.
\end{enumerate}\vs{-1}

Note that all newly added transitions of $N$ (Lines 3--5, 8--9) are deterministic ones. Therefore, for any halting computation path of $M$ on $x$, there always exists its associated computation path of $N$ leading to the corresponding halting state with the same probability.  Thus,  $N$ simulates $M$ on every input and $N$ errs with the same probability as $M$ does. It easily follows by the definition of $N$ that  $|Q'|=2n+2$ and that both the stack alphabet size and the push size are unaltered.
\end{proofof}

%%%
\subsection{Transforming 1ppda's to the Idea Shape Form}\label{sec:ideal shape-lemma}

We next convert an endmarker 1ppda into a specific  ``push-pop-controlled'' form, called an ``ideal shape,'' in which  pop operations always take place while first reading a non-blank input symbol $\sigma$ and then making a (possible) series of pop operations without reading any other input symbol, and push operations add only single symbols without altering any existing stack content.

Any endmarker 1ppda $M$ with $Q$, $\check{\Sigma}$ ($=\Sigma\cup\{\cent,\dollar\}$), $\Gamma$, and $\delta$ is said to be \emph{in an ideal shape} if $M$ makes only moves specified by the following  conditions (1)--(5).
(1) Scanning $\sigma\in \check{\Sigma}$, $M$ preserves the topmost stack symbol (called a \emph{stationary operation}).  (2) Scanning  $\sigma\in \check{\Sigma}$,  $M$ pushes a new symbol $b$ ($\in\Gamma^{(-)}$) without changing any symbol stored in the stack. (3) Scanning  $\sigma\in \check{\Sigma}$, $M$ pops the topmost stack symbol except for $Z_0$. (4) Without scanning any input symbol (i.e., making a $\lambda$-move), $M$ pops the topmost stack symbol except for $Z_0$. (5) The stack operation (4) comes only after either (3) or (4). For later convenience, we refer to these five conditions as the \emph{ideal shape conditions}.

More formally, the above ideal shape conditions are described using $\delta$ by the following lines.
Let $p,q\in Q$, $\sigma\in\check{\Sigma}$, $u\in\Gamma^*$, and $a\in\Gamma$.

\begin{enumerate}\vs{-1}
  \setlength{\topsep}{-2mm}%
  \setlength{\itemsep}{0mm}%
  \setlength{\parskip}{0cm}%

\item If $\delta(q,\sigma,a \mmid p,u)\neq0$, then $u$ must be in $\{\lambda,ba,a\}$ for a certain $b\in\Gamma^{(-)}$.

\item If $\delta(q,\lambda,a \mmid p,u)\neq0$, then $u=\lambda$  and $a\neq Z_0$.

\item If $\delta(q,\sigma,a \mmid p,ba)\neq0$ for a certain $b\in\Gamma^{(-)}$, then $\delta[p,\lambda,b]=0$

\item If $\delta(q,\sigma,a \mmid p,a)\neq0$, then $\delta[p,\lambda,a]=0$.
\end{enumerate}\vs{-1}

\n Notice that Lines 1--2 are related to the stack operations  described above as Conditions (1)--(4) and that Lines 3--4 are related to Condition (5).
It is important to note that Lines 1--4 make $M$ have push size $2$. For no-endmarker 1ppda's, in contrast, we can define the same notion of ``ideal shape'' by requiring Lines 1--4 to hold \emph{only for nonempty inputs}.

Lemma \ref{transition-simple} states that any (no-)endmarker 1ppda can be converted into its error-equivalent (no-)endmarker 1ppda in an ideal shape.

\begin{lemma}\label{transition-simple}{\rm [Ideal Shape Lemma]}
Let $n\in\nat^{+}$. Any $n$-state endmarker 1ppda $M$ with stack alphabet size $m$ and push size $e$ can be converted into another error-equivalent endmarker 1ppda $N$ in an ideal shape with at most $32en^2m^2(2m)^{8enm}$ inner states and at most stack alphabet size $9enm(2m)^{8enm}$.
The above statement is also true for no-endmarker 1ppda's.
\end{lemma}

The error-equivalent result of Lemma \ref{transition-simple} immediately leads to the following corollary.

\begin{corollary}
Lemma \ref{transition-simple} holds also for 1dpda's and 1npda's.
\end{corollary}

The proof of Lemma \ref{transition-simple} partly follows the arguments of Hopcroft and Ullman \cite[Chapter 10]{HU79} and of Pighizzini and Pisoni
\cite[Section 5]{PP15}.
In the proof, the desired conversion of a given 1ppda $M$ into its error-equivalent ideal shape 1ppda $N$ is carried out stage by stage. Firstly, $M$ is converted  into another 1ppda $M_1$ whose $\lambda$-moves are limited only to pop operations.
This $M_1$ is further modified into another 1ppda $M_2$, which satisfies (2)--(3) of the ideal shape conditions together with the extra condition that all $\lambda$-moves are limited to either pop operations or replacements of a topmost stack symbol with a single symbol.
From this $M_2$, we can construct another 1ppda $M_3$, which satisfies (1)--(3) of the ideal shape conditions. Finally, we modify $M_3$ into another 1ppda $M_4$, which further forces the ideal shape condition (5) to meet. This last machine $M_4$ becomes the desired 1ppda $N$.

For readability, we sidestep the proof of Lemma \ref{transition-simple} and continue the transformation of 1ppda's. The actual proof of the lemma will be found in Appendix.

%%%%%

It is important to note that the transformation of a general 1ppda to an error-equivalent 1ppda in an ideal shape is quite costly in terms of stack-state complexity. This high cost affects the overall cost of transforming one endmarker 1ppda to another error-equivalent no-endmarker 1ppda. We refer to Section \ref{sec:final-proof} for the actual cost estimation.

%%%%
\subsection{Usefulness of the Ideal Shape Form}\label{sec:usefulness}

We have introduced  to pushdown automata the notion of ``ideal shape'' form, which controls the behaviors of their push and pop operations. To demonstrate the usefulness of this specific form in handling those machines, we intend to present a simple and quick example of how to apply the ideal shape lemma (Lemma \ref{transition-simple}) for unbounded-error endmarker 1ppda's.

\begin{proposition}\label{unbounded-reverse}
The family of all languages recognized by endmarker 1ppda's with unbounded-error probability is closed under reversal.
\end{proposition}

As far as we know, the following proof is new and worth publishing.

\begin{proofof}{Proposition \ref{unbounded-reverse}}
Let $L$ denote any language and let $M =(Q,\Sigma,\{\cent,\dollar\}, \Gamma,\delta,q_0,Z_0, Q_{acc},Q_{rej})$ be any 1ppda that recognizes $L$ with unbounded-error probability. By applying the ideal shape lemma (Lemma \ref{transition-simple}), we can assume that $M$ is in an ideal shape.
Lemma \ref{postpone-halting} also helps us assume that, on each halting computation path of $M$, $M$ always empties its stack before entering halting states. We aim at constructing another 1ppda $N$ that takes $x^R$ as its input and enters the same type of halting states as $M$ does on the input $x$ with appropriate error probability.
To simplify the following description of $N$'s computation, we further assume that the input tape of $N$ contains $\cent x\dollar$ and that $N$ moves its tape head from right to left, starting at the right endmarker $\dollar$ and ending at the left endmarker $\cent$. See Fig.~\ref{fig:reversing-machine} for the opposite directions of the tape head moves of $M$ and $N$.

A basic idea of the following construction of $N$ is to reverse a computation of $M$ step by step by probabilistically guessing the previous step of $M$ and by verifying the validity of this guessing at the next step. If the guessing is incorrect, then $N$ immediately enters both an accepting state and a rejecting state with an equal probability to ``cancel out''
this incorrect guessing.
It is important to note that this simulation may change the acceptance  and the rejection probabilities significantly, however, the final decision (i.e., acceptance or rejection) of $N$ is always the same as that of $M$ because of the unbounded-error nature of $M$.

%%%%%%
%%%%%%

\begin{figure}[t]
\centering
\includegraphics*[height=2.9cm]{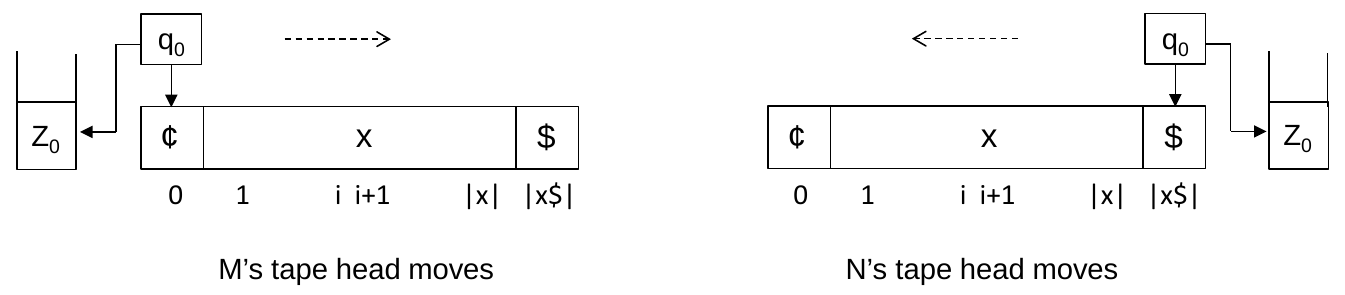}
\caption{Two completely opposite moving directions of the tape heads of $M$ and $N$ described in the proof of Proposition \ref{unbounded-reverse}.}\label{fig:reversing-machine}
\end{figure}

%%%%%
%%%%%

Formally, let $N=(Q',\Sigma,\{\cent,\dollar\}, \Gamma, \delta', q_0,Z_0, Q'_{acc},Q'_{rej})$. Recall that $\check{\Sigma}=\Sigma\cup\{\cent,\dollar\}$. Here, we set $Q'= (Q\times \check{\Sigma} \times Q_{halt}) \cup (Q\times \check{\Sigma} \times Q_{halt}\times \Gamma)$, $Q'_{acc}=\{(q_{acc},\cent,q_{acc})\}$, and   $Q'_{rej}=\{(q_{rej},\cent,q_{rej})\}$. Before defining the desired probabilistic transition function $\delta'$, we first define an intermediate  function $\hat{\delta}$ in Stages (1)--(6) given below.
Let $\sigma,\tau\in \Sigma\cup\{\dollar\}$ and $a,b,c\in\Gamma$.

\s

(1) To reverse $M$'s computation, $N$ needs to simulate both an accepting computation path and a rejecting computation path of $M$ since we do not know whether $M$ will accept or reject. At the first step, $N$ begins in the initial state $q_0$ and splits into two inner states $(q_{acc},\dollar,q_{acc})$ and $(q_{rej},\dollar,q_{rej})$ with an equal probability (Line 1) before  $N$ starts simulating $M$.

\begin{enumerate}\vs{-1}
\item[1.] $\hat{\delta}(q_0,\dollar,Z_0 \mmid (r_0,\dollar,r_0),Z_0) = 1/2$ for all $r_0\in\{q_{acc},q_{rej}\}$.
\end{enumerate}\vs{-1}

\s

(2) Assume that  $M$ either pops $a$ (Line 2) or keeps $a$ intact (Line 3) in inner state $q$ while reading $\sigma$. Notice that, since $M$ is in an ideal shape, $M$ cannot replace $a$ by any other symbol.
If $N$ is in inner state $(p,\tau,r_0)$ for a certain symbol $\tau$, then we force $N$ to change this inner state to $(q,\sigma,r_0)$ exactly when reading $\tau$. The case where $N$ reads a different input symbol other than $\tau$ will be handled in Stage (6). When $M$ pops $a$, $N$ replaces its current top stack symbol, say, $c$ by $ac$ (Line 2). When $M$ maintains $a$, $N$ does the same (Line 3).

\begin{enumerate}\vs{-1}
  \setlength{\topsep}{-2mm}%
  \setlength{\itemsep}{0mm}%
  \setlength{\parskip}{0cm}%

\item[2.] $\hat{\delta}((p,\tau,r_0),\tau,c \mmid (q,\sigma,r_0),ac)=\delta(q,\sigma,a \mmid p,\lambda)$.

\item[3.] $\hat{\delta}((p,\tau,r_0),\tau,a \mmid (q,\sigma,r_0),a) = \delta(q,\sigma,a \mmid p,a)$.
\end{enumerate}\vs{-1}

Assuming that $\sigma=x_{(i)}$ and $\tau=x_{(i+1)}$,
if $(q,i,aw)\vdash^{>0}_{M}(p,i+1,w)$ holds for stack content $aw$, $N$ changes its corresponding configuration  $((p,\tau,r_0),i+1,cw')$ satisfying $w=cw'$ into another configuration $((q,\sigma,r_)),i,acw')$ by applying Line 2.
When $(q,i,aw)\vdash^{>0}_{M}(p,i+1,bw)$, in contrast, $N$'s configuration  $((p,\tau,r_0),i+1,aw)$ is changed into $((q,\sigma,r_0),i,aw)$ by Line 3.

\s

(3) In the case where $M$ pushes a symbol $b$ onto $c$ while reading $\sigma$, $N$ first pops $b$, guesses such a stack symbol $c$ (Line 4), and deterministically checks whether the top stack symbol is indeed $c$ by making a $\lambda$-move (Line 5).

\begin{enumerate}\vs{-1}
  \setlength{\topsep}{-2mm}%
  \setlength{\itemsep}{0mm}%
  \setlength{\parskip}{0cm}%

\item[4.] $\hat{\delta}((p,\tau,r_0),\tau,b \mmid (q,\sigma,r_0,c),\lambda) =\delta(q,\sigma,c \mmid p,bc)$.

\item[5.] $\hat{\delta}((q,\sigma,r_0,c),\lambda,c\mmid (q,\sigma,r_0),c)=1$.
\end{enumerate}\vs{-1}

Assuming that $\sigma=x_{(i)}$ and $\tau=x_{(i+1)}$, if  $(q,i,cw)\vdash^{>0}_{M}(p,i+1,bcw)$, then $((p,\tau,r_0),i+1,bcw)$ is changed to $((q,\sigma,r_0,c),i,cw)$ (Line 4), followed by $((q,\sigma,r_0),i,cw)$ (Line 5). It may be possible that the sum  $\sum_{c\in\Gamma} \hat{\delta}(q,\sigma,c\mmid p,bc)$ exceeds $1$.

\s

(4) In the case where $M$ makes a $\lambda$-move from inner state $q$ to $p$, $M$ must pop the top stack symbol (Line 6) due to the ideal shape conditions. In this case, $N$ takes a step similar to Line 2.

\begin{enumerate}\vs{-1}
\item[6.] $\hat{\delta}((p,\tau,r_0),\lambda,c \mmid (q,\tau,r_0),ac) =\delta(q,\lambda,a \mmid p,\lambda)$.
\end{enumerate}\vs{-1}

\s

(5) When $N$ reaches $\cent$ in inner state $(q_0,\cent,r_0)$, $N$ deterministically enters the inner state $(r_0,\cent,r_0)$ by executing Line 7.

\begin{enumerate}\vs{-1}
\item[7.] $\hat{\delta}((q_0,\cent,r_0),\cent,c \mmid (r_0,\cent,r_0),c) = 1$.
\end{enumerate}\vs{-1}

\s

(6) For any other tuples $(q,\tau,r_0,\sigma,c)$, we force $N$ to enter an accepting state and a rejecting state with an equal probability (Line 8) so that we eliminate any influence of computation paths that alter the final decision of $N$.

\begin{enumerate}\vs{-1}
\item[8.] $\hat{\delta}((q,\tau,r_0),\sigma,c \mmid \hat{p},c) = 1/2$ if $\sigma\neq\tau$ and  $\hat{p}\in\{q_{acc},q_{rej}\}$.
\end{enumerate}\vs{-1}

This completes the definition of $\hat{\delta}$. Notice that $\hat{\delta}$ may not be a true probabilistic transition function. Therefore, we further modify $\hat{\delta}$ into a proper probabilistic transition function in the following fashion.
Let
\[
\alpha_0 =\max\{\hat{\delta}[(p,\tau,r_0),\tau,b], \hat{\delta}[(p,\tau,r_0),\lambda,b], \hat{\delta}[(p,\tau,r_0,c),\lambda,c] \mid r_0\in\{q_{acc},q_{rej}\}, p\in Q, b,c\in\Gamma, \tau\in\check{\Sigma}\}.
 \]
If $\alpha_0\leq 1$, then we set $\delta'$ to be $\hat{\delta}$. Otherwise, we set $\delta'((p,\tau,r_0),u,b\mmid (q,\sigma,r_0),v)$ to be $\hat{\delta}((p,\tau,r_0),u,b\mmid (q,\sigma,r_0),v)/\alpha_0$ for $u\in\Sigma_{\lambda}$ and $v\in\Gamma^{\leq 2}\cup\{\lambda\}$. If $\delta'[(p,\tau,r_0),u,b]<1$ at this point, then we add extra transitions of the form $\delta'((p,\tau,r_0),u,b\mmid \hat{p},b) = \frac{1}{2}(1-\delta'[(p,\tau),u,b])$ for any $\hat{p}\in\{q_{acc},q_{rej}\}$ so that $\delta'[(p,\tau,r_0),u,b]$ precisely becomes $1$.

We can prove by the length of computation subpaths that, for each halting computation path of $M$ with probability $\gamma$, there exists a corresponding computation path of $N$ moving the tape head from right to left  with probability $\gamma\cdot \max\{\alpha_0,1/\alpha_0\}$.
Since the value $\alpha_0$ is fixed for both accepting and rejecting computation paths, it follows that $p_{M,acc}(x)>p_{M,rej}(x)$ iff $p_{N,acc}(x^R) > p_{N,rej}(x^R)$. Therefore, $N$ correctly recognizes $L^R$ with unbounded-error probability.
\end{proofof}

%%%%%%
%%%%%%
\section{From Endmarker 1ppda's to No-Endmarker 1ppda's}\label{sec:deleting-endmarker}

Following the preparatory steps of Section \ref{sec:modifications}, let us prove Proposition \ref{deleting-endmarker}. The proof of this proposition is composed of several key ingredients, each of which uses a different idea to gradually place more restrictions on 1ppda's. In particular, the proof will proceed in three stages by firstly removing $\cent$ (Section \ref{sec:remove-left}), secondly removing a final series of $\lambda$-moves after reading $\dollar$ (Section \ref{sec:remove-empty}), and finally removing $\dollar$ (Section \ref{sec:remove-endmarker}). In Section \ref{sec:final-proof}, we will combine those transformations with the desired estimation of the stack-state complexity and complete the proof of Proposition \ref{deleting-endmarker}.

As far as there is no confusion, we loosely refer to a string (including endmarker(s)) written on an input tape as an \emph{input string}. Therefore, when $\cent x\dollar$ and $x\dollar$ are the entire strings written on the input tape, the first input symbols of them are $\cent$ and $x_{(1)}$, receptively.

%%%
\subsection{Removing the Left Endmarker}\label{sec:remove-left}

We begin by removing the left endmarker $\cent$ from the 1ppda's input tape. What we need here is to construct a 1ppda without $\cent$, which correctly simulates a  given endmarker 1ppda. If we had made the first set of moves ``$\lambda$-moves'', then we could have easily simulated the given machine. By our definition of 1ppda's given
in Section \ref{sec:first-formulation}, however,  no 1ppda can start with a $\lambda$-move, and thus we require a different way to simulate any 1ppda with $\cent$.

For clarity, we intend to call by a \emph{no-left-endmarker 1ppda} any 1ppda whose input tape uses only the right endmarker $\dollar$ with no use of the left endmarker $\cent$. Thus, the tape head starts at the leftmost input symbol,  instead of $\cent$, written on the input tape.
As a special case, when the input is $\lambda$, the leftmost input symbol becomes $\dollar$. Fig.~\ref{fig:no-left-endmarker} illustrates the difference between an endmarker 1ppda and a no-left-endmarker 1ppda.

Here, we assume a given endmarker 1ppda to be in an ideal shape. For the subsequent transformation of 1ppda's in Section \ref{sec:remove-empty}, we intend to construct in Lemma \ref{remove-left} its error-equivalent 1ppda satisfying  the ideal shape conditions.

%%%%%%
%%%%%%

\begin{figure}[t]
\centering
\includegraphics*[height=3.0cm]{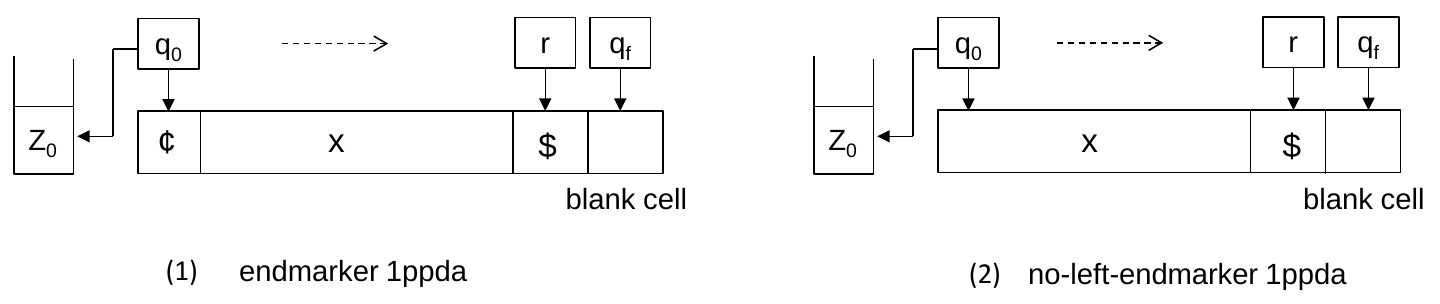}
\caption{(1) An endmarker 1ppda starts by reading the left endmarker $\cent$ in the initial state $q_0$. The machine may make a series of consecutive $\lambda$-moves even after reading $\dollar$ and eventually enters halting states, say, $q_f$. (2) A no-left-endmarker 1ppda  starts by  reading the leftmost  symbol of input $x$ in the initial state $q_0$. The machine may make a series of $\lambda$-moves after processing $\dollar$ and eventually enters the halting states $q_f$.}\label{fig:no-left-endmarker}
\end{figure}

%%%%%
%%%%%

\begin{lemma}\label{remove-left}
Let $M$ denote any $n$-state endmarker 1ppda in an ideal shape  over input alphabet $\Sigma$ with stack alphabet size $m$.
There is a 1ppda $N$ that uses no left endmarker, has $n(|\Sigma|+2)+1$ states and $m$ stack alphabet size, and is error-equivalent to $M$. Moreover, $N$ is in an ideal shape.
\end{lemma}

To obtain a precise stack-state complexity of transformation given in Lemma \ref{remove-left}, we provide an explicit construction of the desired machine $N$.
A basic idea of the construction of $N$ is to simulate in a single move the first two moves of $M$ made when reading $\cent$ and the next input symbol.
To ensure the ideal shape conditions, whenever $M$ consecutively pushes twice, we instead push one symbol and remember the top stack symbol in the form of inner states.

\vs{-2}
\begin{proofof}{Lemma \ref{remove-left}}
We begin with an arbitrary 1ppda $M$ of the form $(Q,\Sigma,\{\cent,\dollar\}, \Gamma,\Theta_{\Gamma},\delta, q_0,Z_0, Q_{acc},Q_{rej})$ equipped with the two endmarkers with $|Q|=n$ and $|\Gamma| =m$.
Let us recall the notations $\Gamma^{(-)}$, $Q_{halt}$,   $\Sigma_{\lambda}$, $\check{\Sigma}_{\lambda}$  from Section \ref{sec:numbers} and remember that $\delta$ maps $(Q-Q_{halt})\times \check{\Sigma}_{\lambda}\times \Gamma \times Q \times \Theta_{\Gamma}$ to $[0,1]$.  We further assume that $M$ is in an ideal shape.
We then want to construct its error-equivalent 1ppda $N$ that uses no left endmarker $\cent$. We express  $N$ as $(Q',\Sigma,\{\dollar\}, \Gamma',\Theta_{\Gamma'}, \delta',\hat{q}_0, Z_0,Q'_{acc},Q'_{rej})$ and hereafter describe each component of $N$ in detail. Let $x=x_{(1)}x_{(2)}\cdots x_{(n)}$ denote any input string of length $|x|$ over $\Sigma$.
We set $Q'= Q\cup\{\bar{q}_0\} \cup (Q\times \Sigma_{\dollar})$ and $\Gamma'=\Gamma$.
Clearly, $|Q'|=n(|\Sigma|+2)+1$ and $|\Gamma'|=m$ follow.

The probabilistic transition function $\delta'$ is formally given in Stages (1)--(2) described below.

\s

(1)
By adopting the second formalism of Section \ref{sec:second-formulation}, at the very first step, $N$ should read the leftmost symbol $x_{(1)}$ of the input string $x\dollar$ (written on the input tape) without making any $\lambda$-move.

Assuming $x\neq\lambda$,
at the first step taken while reading $\cent$, $M$ either pushes a stack symbol, say, $a$ or preserves $Z_0$. Since $M$ is in an ideal shape, there is no $\lambda$-move of $M$ immediately after reading $\cent$.
After $M$ reads $\cent$ at the first step, $M$'s tape head should move to the next  input symbol $x_{(1)}$.  Hence, at the second step, $M$ further reads  $x_{(1)}$.
Since $N$ begins by reading the leftmost input symbol $x_{(1)}$ instead of $\cent$, we need to wait until $M$ completely reads $\cent x_{(1)}$. Our initial attempt is to combine the first two steps of $M$ to form a single move of $N$.
For this purpose, we define a supportive function $\bar{\delta}$, which is given by the following equations. Let $\sigma\in\Sigma$, $a,b\in\Gamma^{(-)}$, and $p\in Q$.

\begin{enumerate}\vs{-1}
  \setlength{\topsep}{-2mm}%
  \setlength{\itemsep}{0mm}%
  \setlength{\parskip}{0cm}%

\item[$(a)$] $\bar{\delta}(\hat{q}_0,\sigma, Z_0 \mmid (p,0,0),Z_0) = \sum_{p'\in Q} [\delta(q_0,\cent,Z_0\mmid p',Z_0) \delta(p',\sigma,Z_0\mmid p,Z_0)]$.

\item[$(b)$] $\bar{\delta}(\hat{q}_0,\sigma, Z_0 \mmid (p,0,a),Z_0) = \sum_{p'\in Q} [\delta(q_0,\cent,Z_0\mmid p',Z_0) \delta(p',\sigma,Z_0\mmid p,aZ_0)]$.

\item[$(c)$] $\bar{\delta}(\hat{q}_0,\sigma, Z_0 \mmid (p,a,0),Z_0) = \sum_{p'\in Q} [\delta(q_0,\cent,Z_0\mmid p', aZ_0) \delta(p',\sigma,a\mmid p,a)]$.

\item[$(d)$] $\bar{\delta}(\hat{q}_0,\sigma, Z_0 \mmid (p,a,b),Z_0) = \sum_{p'\in Q} [\delta(q_0,\cent,Z_0\mmid p',aZ_0) \delta(p',\sigma, a\mmid p, ba)]$.

\item[$(e)$] $\bar{\delta}(\hat{q}_0,\sigma, Z_0 \mmid (p,a,-1),Z_0) = \sum_{p'\in Q} [\delta(q_0,\cent,Z_0\mmid p', aZ_0) \delta(p',\sigma,a\mmid p,\lambda)]$.
\end{enumerate}\vs{-1}

The machine $N$ then makes the combined move of $M$ by updating $N$'s stack content (Lines 1--3).  In particular, if $M$ pushes two symbols $ba$ simultaneously, then we push only $a$ but  remember the top stack symbol $b$ as an inner state of the form of $(q,b)$  (Line 3). This is necessary because, if  we replace
Line 3 by a much simpler transition ``$\delta'(\hat{q}_0,\sigma,Z_0\mmid p, baZ_0) = \bar{\delta}(\hat{q}_0,\sigma,Z_0\mmid  (p,a,b),Z_0)$'', then we obviously violate the ideal shape conditions.

The desired probabilistic transition function $\delta'$ is defined with the help of a supportive function $\bar{\delta}$ in the following simple way.

\begin{enumerate}\vs{-1}
  \setlength{\topsep}{-2mm}%
  \setlength{\itemsep}{0mm}%
  \setlength{\parskip}{0cm}%

\item[1.] $\delta'(\hat{q}_0,\sigma,Z_0\mmid p,Z_0) = \bar{\delta}(\hat{q}_0,\sigma,Z_0\mmid  (p,0,0),Z_0) + \sum_{a\in \Gamma^{(-)}}\bar{\delta}(\hat{q}_0,\sigma,Z_0\mmid  (p,a,-1),Z_0)$.

\item[2.] $\delta'(\hat{q}_0,\sigma,Z_0\mmid p,aZ_0) = \bar{\delta}(\hat{q}_0,\sigma,Z_0\mmid  (p,a,0),Z_0) + \bar{\delta}(\hat{q}_0,\sigma,Z_0\mmid  (p,0,a),Z_0)$.

\item[3.] $\delta'(\hat{q}_0,\sigma,Z_0\mmid (p,b), aZ_0) = \bar{\delta}(\hat{q}_0,\sigma,Z_0\mmid  (p,a,b),Z_0)$.
\end{enumerate}\vs{-1}

Whenever $x$ is $\lambda$, by contrast, $N$ must read $\dollar$ at the first step. In this case, we simply replace $\sigma$  in Lines 1--3 by $\dollar$.

%%%

\s

(2) Once $M$'s tape head moves away from $x_{(1)}$, since there is no following $\lambda$-move, $M$ must read the next input symbol $x_{(2)}$. Following Line 3, if $M$ further pushes  symbol $c$, then we push $b$ instead and update our memory of $(q,b)$ to $(p,c)$ (Line 4). Similarly, when $M$ preserves $b$, we only update the memory (Line 5). In the case of $M$'s popping $b$, if there is no following $\lambda$-move, then we simply remove the data regarding $b$ from our memory (Line 6); otherwise, we additionally pop $a$ (Line 7). As long as our memory has an inner state of the form $(q,b)$, we continue applying Lines 4--7.
Once the complete removal of all the data of the  from $(q,b)$ is done, we simply mimic $M$'s behaviors step by step (Line 8).
Let $\xi\in\Sigma\cup\{\dollar\}$, $\eta\in\Sigma_{\lambda}$, $a,b,c\in\Gamma^{(-)}$, $d\in\Gamma$, $p,q\in Q$, and $w\in\Theta_{\Gamma}$.  

\begin{enumerate}\vs{-1}
  \setlength{\topsep}{-2mm}%
  \setlength{\itemsep}{0mm}%
  \setlength{\parskip}{0cm}%

\item[4.] $\delta'((q,b),\xi,a\mmid (p,c), ba) = \delta(q,\xi,b\mmid p,cb)$.

\item[5.] $\delta'((q,b),\xi,a\mmid (p,b), a) = \delta(q,\xi,b\mmid p,b)$.

\item[6.] $\delta'((q,b),\xi,a\mmid p, a) = \delta(q,\xi,b\mmid p,\lambda)$ if there is no following $\lambda$-move.

\item[7.] $\delta'((q,b),\xi,a\mmid p, \lambda) = \sum_{p'\in Q} [\delta(q,\xi,b\mmid p',\lambda) \delta(p',\lambda,a\mmid p,\lambda)]$.

\item[8.] $\delta'(q,\eta,d \mmid p,w) = \delta(q,\eta,d \mmid p,w)$.
\end{enumerate}\vs{-1}

\n The above phrase ``no following $\lambda$-move'' means that $\sum_{p'\in Q}\delta(p,\lambda,a\mmid p',\lambda)=0$.

This completes the construction of $N$. The above construction guarantees the ideal shape form of $N$.
\end{proofof}

%%%
\subsection{Removing the Final Series of $\lambda$-Moves}\label{sec:remove-empty}

Next, let us focus on a no-left-endmarker 1ppda $M = (Q, \Sigma,\{\dollar\}, \Gamma,\Theta_{\gamma},\delta, q_0,Z_0,Q_{acc},Q_{rej})$  in an ideal shape.
Clearly, if $M$ enters halting states before reading the right endmarker $\dollar$ for all but finitely many inputs, then there is no trouble removing $\dollar$ by remembering all those exceptional inputs  and the result (i.e., acceptance or rejection) of $M$ using inner states.
Hence, we hereafter assume that $M$ reaches $\dollar$ on infinitely many inputs.
Recall that, in our endmarker model, a 1ppda is permitted to make a final series of $\lambda$-moves (i.e., a series of $\lambda$-moves just after reading the right endmarker $\dollar$).
In Lemma \ref{final-lambda-move} stated below, we wish to remove this final series of $\lambda$-moves and force $M$ to halt immediately after reading $\dollar$.  Fig.~\ref{fig:stack-history} illustrates a final series of $\lambda$-moves and its removal.

In the statement of the lemma, however, we need to slightly relax the ideal shape conditions (1)--(5)  of Section \ref{sec:ideal shape-lemma} and introduce a new notion of ``almost ideal shape''. A 1ppda $N$ is said to be \emph{in an almost ideal shape} if there exists a subset $\Gamma_{Z_0}$ of $\Gamma^{(-)}$ for which (a)  the conditions (1)--(5) hold for all the steps of $N$ except for the first step, (b) at the first step, $N$ in the initial state $q_0$ reads a tape symbol $a$ and pushes either $\bar{Z}_0$ or $b \bar{Z}_0$ (for $b\in\Gamma^{(-)} - \Gamma_{Z_0}$) for a certain stack symbol $\bar{Z}_0\in \Gamma_{Z_0}$, (c) during a computation of $N$, all symbols in $\Gamma_{Z_0}$ are treated as if they are new ``bottom markers'' (i.e., they are not popped or pushed again after the first step).

\begin{lemma}\label{final-lambda-move}
Let $M$ denote any $n$-state no-left-endmarker 1ppda with stack alphabet size $m$  in an ideal shape.  Along each halting computation path, assume that $M$ enters a halting state only after emptying the stack.
There exists an error-equivalent no-left-endmarker 1ppda $P$ in an almost ideal shape that has $2\bar{n}(\bar{n}2^{\bar{n}}+1)+1$ states and stack alphabet size $m^2 2^{\bar{n}}$ but makes no $\lambda$-moves after reading the right endmarker $\dollar$, where $\bar{n}=n(m+1)$.
\end{lemma}

Using the inequality $r+1 < 2r$ for all real numbers $r>1$, we obtain $2\bar{n}(\bar{n}2^{\bar{n}}+1)+1  \leq 16n^2m^22^{2nm}$ and $m^2 2^{\bar{n}}\leq m^22^{2nm}$. These bounds will be used in Section \ref{sec:final-proof}.

Lemma \ref{final-lambda-move} is the most essential part of our proof of Proposition \ref{deleting-endmarker}.
The subsequent proof of the lemma is inspired by \cite[Lemma 2]{KGF97}, in which
Ka\c{n}eps, Geidmanis, and Freivalds demonstrated how to convert each bounded-error ``unary''  1ppda into an equivalent bounded-error ``inputless''  1ppda, where a \emph{unary machine} uses only a unary input alphabet.
However, our setting (including our assumption) is quite different from \cite[Lemma 2]{KGF97}, and therefore we intend to provide the complete proof of the lemma.
A key idea of Ka\c{n}eps et. al. relies on a result of Nasu and Honda \cite[Section 6]{NH68} concerning a certain property of \emph{one-way probabilistic finite automata} (or 1pfa's) on reversed inputs. In essence, for any given 1pfa $M$ with no $\lambda$-move, Nasu and Honda constructed another 1pfa $N$ that accepts $x^R$ whenever $M$ accepts $x$, with the ``same''  probability, as long as we carefully choose which inner states of $N$ to start and terminate.

The ideal shape conditions of $M$ force each final series of $\lambda$-moves to be composed only of pop operations, and thus $M$ constantly reduces the stack size after reading $\dollar$.
Such a series of pop operations (of the form $\delta_{M}(q,\lambda,\sigma \mmid p,\lambda)$)  can be seen as a series of probabilistic transitions consuming all stack symbols one by one from the topmost symbol to the bottom symbol $Z_0$. In other words, this process is completely described in terms of an appropriate 1pfa that takes a stack content as its input and makes no $\lambda$-moves.
Such a 1pfa $N$ can be expressed as
$(Q,\Gamma^{(-)},\{Z_0\},\delta_N)$ \emph{with neither initial state nor halting states}, where $\Gamma^{(-)}$ is used as an input alphabet, $Z_0$ is treated as a new right endmarker (but no left endmarker), and $\delta_{N}(q,\sigma \mmid p) = \delta_{M}(q,\lambda,\sigma \mmid p,\lambda)$ for any tuple $(p,q,\sigma)$ with $\sigma\neq\lambda$.
Notice that $\Gamma^{(-)}\cup\{Z_0\}$ coincides with $\Gamma$.
We force $N$ to run until $Z_0$ (as an input symbol) is completely read. For our convenience, we call such a machine a \emph{free 1pfa}.
Given such a free 1pfa $N$, the notation  $p_{N}(q,w,p)$ denotes the probability that $N$ starts in inner state $q$, reads $w$ completely, and then enters inner state $p$.
Formally, for any string $w=w_1w_2\cdots w_n\in(\Gamma^{(-)})^*\cup (\Gamma^{(-)})^*Z_0$, $p_{N}(q,w,p)$ is defined to be $\sum_{p_2,p_3,\ldots,p_n\in Q} ( \prod_{i\in[n]} \delta_N(p_i,w_i \mmid p_{i+1}) )$ with $p_1=q$ and $p_{n+1}=p$.

The acceptance/rejection probability of $M$ on input $x$ is given by the probability distribution produced by $M$ after completely reading $x\dollar$, without making $\lambda$-moves after reading $\dollar$, multiplied by the probability of $N$'s processing (as its own input) the stack content generated by $M$ on $x\dollar$. However, to simulate $N$ properly, we need to read the stack content of $M$ \emph{from the bottom symbol to the top symbol}, which is the ``reverse'' of the input given to $N$. What we need for the construction of the desired 1ppda $P$ stated in Lemma \ref{final-lambda-move} is another free 1pfa $K$ ``mimicking'' $N$ on the input $x$ but $K$ must read the reverse of $x$.

We thus wish to introduce a supportive lemma, stated as Lemma \ref{Nasu-Honda}. This lemma ensures that we can construct such a free 1pfa $K$ that ``mimics'' the behavior of $N$ by reading a ``reversed'' input if we choose appropriate inner states to start and terminate.

%%%%%%
%%%%%%

\begin{figure}[t]
\centering
\includegraphics*[height=3.6cm]{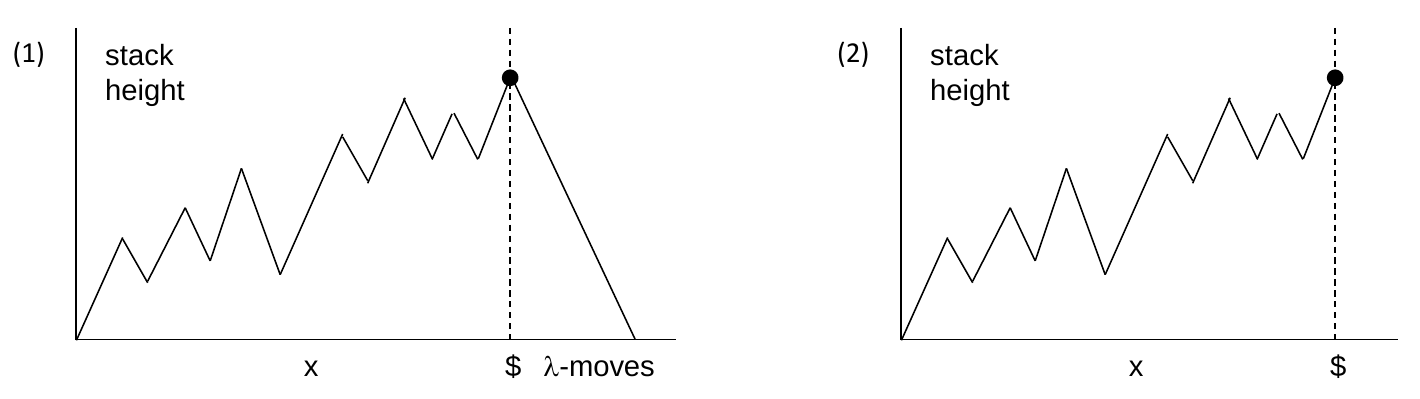}
\caption{(1) After reading $\dollar$, the underlying 1ppda makes a series of consecutive $\lambda$-moves to empty the stack. Once the stack becomes empty, the 1ppda enters halting states. (2) The 1ppda enters halting states at reading $\dollar$, and thus the 1ppda makes no $\lambda$-move after reading $\dollar$.}\label{fig:stack-history}
\end{figure}

%%%%%
%%%%%

\begin{lemma}\label{Nasu-Honda}
Take $N$ and $p_{N}$ that are constructed from $M$ as described above. There exists a free 1pfa  $K=(Q_K,\Gamma^{(-)},\{Z_0\},\delta_K)$ with  $|Q_K|=2^{|Q|}$ that satisfies the following: for any pair $p,q\in Q$, there exist an inner state $s(p)\in Q_K$ and a subset $T(q)\subseteq Q_K$ satisfying $p_{N}(q,w,p) = p_K(s(p),w^R,T(q))$ ($=\sum_{r\in T(q)}p_{K}(s(p),w^R,r)$)  for any string $w\in (\Gamma^{(-)})^* \cup (\Gamma^{(-)})^*Z_0$, where $p_{K}$ is defined from $K$ similarly to $p_{N}$.
\end{lemma}

Our proof of Lemma \ref{Nasu-Honda} closely follows the proof of Nasu and Honda \cite[Theorem 6.1]{NH68} in order to construct the desired $K$.  For this purpose, we first reformulate a computation of $N$ in a more traditional way using a set of \emph{transition matrices} instead of the transition function $\delta_N$. Given a symbol $\sigma\in\Gamma^{(-)}$, we write $U_{\sigma} = (a_{qp}^{(\sigma)})_{q,p\in Q}$ for a $|Q|\times|Q|$ transition matrix obtained by setting each $(q,p)$-entry $a_{qp}^{(\sigma)}$ to be the real number
$\delta_N(q,\sigma\mmid p)$.
Since $N$ is a free 1pfa, $U_{\sigma}$ is a stochastic matrix.\footnote{An $n\times n$ real matrix $U=(a_{ij})_{i,j\in[n]}$ is said to be \emph{stochastic} if the sum of each row equals $1$; that is, for each $i\in[n]$, $\sum_{j=1}^{n}a_{ij}=1$ holds.}
Given an inner state $q\in Q$, $\pi_q$ denotes the row vector in which each entry indexed by $r$ is $1$ if $r=q$ and $0$ otherwise.
It then follows that $a_{qp}^{(\sigma)}$ equals $\pi_qU_{\sigma}\pi_p^T$. For any string $w\in (\Gamma^{(-)})^*$, $p_N(q,w,p)$ is therefore expressed as $\pi_q U_{w}\pi_p^T$, where $U_w$ indicates the matrix multiplication $U_{w_1}U_{w_2}\cdots U_{w_n}$ when  $w=w_1w_2\cdots w_n$ with  $w_1,w_2,\ldots,w_n\in\Gamma^{(-)}$. Notice that, since $p_N(q,w,p)$ is merely a real number,  $p_N(q,w,p)$ is also expressed as $(\pi_q U_w \pi_p^T)^T$, which is simply $\pi_p U_w^T \pi_q^T$.
This means that the transpose $U_w^T$  of $U_{w}$  ``simulates'' $N$ on $w$ by reading its  reverse $w^R$.
It thus suffices to construct a stochastic matrix $V_w$ that produces the same acceptance/rejection distribution as  $U_w^T$ does. However, since the transpose $U_w^T$ may not be stochastic in general, we need to expand the original vector space used for $U_w^T$ to a much larger vector space to construct $V_{w}$.
With this $V_w$, a construction similar to Nasu and Honda's makes it possible to build the desired $K$ from $N$.

\vs{-2}
\begin{proofof}{Lemma \ref{Nasu-Honda}}
Let $M= (Q,\Sigma,\{\dollar\}, \Gamma,\Theta_{\Gamma}, \delta,q_0,Z_0, Q_{acc}, Q_{rej})$ denote a no-left-endmarker 1ppda in an ideal shape and let $N=(Q,\Gamma^{(-)},\{Z_0\},\delta_N)$ be its associated free 1pfa described above.
In what follows, we use the notation $[\![x\in A]\!]$ to denote the \emph{truth value} of the statement ``$x\in A$''; namely, $[\![x\in A]\!]$ equals $1$ if $x\in A$, and $0$ otherwise. Let $n=|Q|$.
For readability, without loss of generality, we identify $Q$ with the integer set $[n]$ with $q_0=1$. We then set $Q_K=\PP(Q)$. Clearly,  $|Q_K|=|\PP(Q)|=2^n$ follows.
We thus express $Q_K$ as $\{Q_1,Q_2,\ldots,Q_{2^n}\}$ with $Q_1=\setempty$ and $Q_{2^n}=Q$.  For convenience, we identify $Q_K$ with $[2^n]$ by assigning $k$ to $Q_k$ in a later argument.
For each symbol $\sigma\in\Gamma^{(-)}$,
we abbreviate $\delta_{N}(i,\sigma \mmid j)$ as  $\alpha^{(\sigma)}_{ij}$ for each pair $i,j\in Q$ and consider the $n\times n$ transition matrix $U_{\sigma}=(\alpha_{ij}^{(\sigma)})_{i,j\in Q}$ of $N$ associated with each symbol $\sigma$. Notice that $U_{\sigma}$ is a stochastic matrix.
Moreover, we define $\pi_i=(\pi_{i1},\pi_{i2},\ldots,\pi_{in})$ by setting  $\pi_{ii}=1$ and $\pi_{ij}=0$ for any distinct pair $i,j\in Q$.

As noted prior to this proof, we obtain $p_N(i,w,j) = \pi_i U_{w}\pi_j^T = (\pi_i U_{w}\pi_j^T)^T =  \pi_jU_w^T\pi_i^T$. Since $U_w^T = U_{w_n}^T U_{w_{n-1}}^T \cdots U_{w_1}^T$, it suffices to define the desired stochastic matrix $V_{\sigma}$ so that it can ``simulate'' $U_{\sigma}^T$ for every symbol $\sigma$.
In what follows, we fix $\sigma$ arbitrarily.
Given any index $j\in Q$, we express the $j$th row  $(\alpha_{1j}^{(\sigma)},\alpha_{2j}^{(\sigma)}, \ldots,\alpha_{nj}^{(\sigma)})$ of $U_{\sigma}^T$ as $\alpha_j(\sigma)$.

Hereafter, we intend to construct a $2^n\times2^n$ matrix $V_{\sigma}=(\gamma_{kl}^{(\sigma)})_{k,l\in[2^n]}$, which can ``simulate''  $U_{\sigma}^T$.  For each index $k\in[2^n]$, let $\gamma_k(\sigma)$ denote the $k$th row of $V_{\sigma}$ and assume that $\gamma_k(\sigma)$ has the form $(\gamma_{k1}^{(\sigma)}, \gamma_{k2}^{(\sigma)}, \ldots, \gamma_{k2^n}^{(\sigma)})$.
To define $\gamma_k(\sigma)$, we first need to consider the supportive vector $\beta_k(\sigma) = \sum_{j\in Q_k}\alpha_j(\sigma)$.
For later use, we write $(\beta_{k1}^{(\sigma)}, \beta_{k2}^{(\sigma)},\ldots,\beta_{kn}^{(\sigma)})$ for  $\beta_k(\sigma)$.  It then follows that $0\leq \beta_{kl}^{(\sigma)}\leq 1$ for any $l\in[2^n]$ since $\beta_{kl}^{(\sigma)}=\sum_{j\in Q_k}\alpha^{(\sigma)}_{lj}$ and $\sum_{j=1}^{n}\alpha^{(\sigma)}_{lj}=1$.

As our preparation, we introduce three extra notations for any nonnegative real vector $\beta=(\beta_1,\beta_2,\ldots,\beta_n)$:
(1)  $\min(\beta)$ denotes the minimum positive value of any entry of the vector $\beta$ (if one exists), namely, $\min(\beta) = \min\{\beta_{j}\mid j\in[n], \beta_{j}>0\}$, (2) $\max(\beta)$ denotes the maximum value of any entry of $\beta$, and
(3) $e_k$ denotes the $\{0,1\}$-vector $([\![1\in Q_k]\!],[\![2\in Q_k]\!], \ldots,[\![n\in Q_k]\!])$, provided that $Q$ is identified with the integer set $[n]$.
In particular, since $Q_1=\setempty$ and $Q_{2^n}=Q$, we obtain $e_1=(0,0,\ldots,0)$ and $e_{2^n}=(1,1,\ldots,1)$.

Fix $k\in[2^n]$ arbitrarily.
We then construct a series of supportive vectors $\hat{\beta}_{k}^{(1)}(\sigma), \hat{\beta}_{k}^{(2)}(\sigma), \ldots$ in the following inductive way. Initially, we set $\hat{\beta}_{k}^{(1)}(\sigma)$ to be $\beta_k(\sigma)$. Given an index $i\in \nat^{+}$,  $\hat{\beta}_k^{(i)}(\sigma)$ is assumed to have already been defined. If $\hat{\beta}_k^{(i)}(\sigma)={\bf 0}$ (the zero vector), then we stop the construction process. Assuming otherwise, let $\hat{\beta}_k^{(i)}(\sigma) = (\tilde{\beta}_1^{(i)},\tilde{\beta}_2^{(i)},\ldots, \tilde{\beta}_n^{(i)})$. 
Since $\{j\in[n]\mid \tilde{\beta}^{(i)}_j>0\}$ is not empty, we can choose an index $m(i)\in [2^n]$ for which $Q_{m(i)}$ equals $\{j\in[n]\mid \tilde{\beta}_j^{(i)}>0\}$.
Notice that such an index always exists. Since $\min(\hat{\beta}^{(i)}_{k}(\sigma))$ exists, we can set  $\gamma_{km(i)}^{(\sigma)}= \min(\hat{\beta}_k^{(i)}(\sigma))$.
We then define  $\hat{\beta}_k^{(i+1)}(\sigma) = \hat{\beta}_k^{(i)}(\sigma) - \gamma_{km(i)}^{(\sigma)}\cdot e_{m(i)}$.
If the construction process terminates after $t$ steps, then
$1\leq t\leq n$ follows since each step eliminates at least one entry, which is the smallest positive number.  Notice that $0\leq \max(\beta_k(\sigma))\leq 1$ because $0\leq \beta^{(\sigma)}_{kl}\leq 1$ for all $l\in[2^n]$.
After the process terminates, we set $\gamma_{k1}^{(\sigma)}=1-\max (\beta_{k}(\sigma))$ and $\gamma_{kl}^{(\sigma)}=0$ for any $l\in[2^n]$ that are undefined so far.  Since $\sum_{i\in[t]}\min (\hat{\beta}_k^{(i)}(\sigma)) = \sum_{j\in[n]} \beta_{kj}^{(\sigma)} = \max( \beta_k(\sigma))$, it  follows that $\sum_{l\in[2^n]} \gamma_{kl}^{(\sigma)} = \gamma_{k1}^{(\sigma)} + \sum_{i\in[t]}\min (\hat{\beta}_k^{(i)}(\sigma)) = \gamma_{k1}^{(\sigma)} + \max( \beta_k(\sigma)) = 1$.
This ensures the stochasticity of  $V_{\sigma}$. We further define $\hat{\pi}_k = (\hat{\pi}_{k1},\hat{\pi}_{k2},\ldots, \hat{\pi}_{k2^n})$ by setting $\hat{\pi}_{kk}=1$ and $\hat{\pi}_{kl}=0$ for any distinct pair $k,l\in[2^n]$.

Finally, we define the desired free 1pfa $K=(Q_K,\Gamma^{(-)},\{Z_0\},\delta_K)$ by assigning $\gamma^{(\sigma)}_{kl}$ to $\delta_K(Q_k,\sigma \mmid Q_l)$ for any pair $k,l\in [2^n]$.
This machine $K$ ensures that, for any string $w=w_1w_2\cdots w_n$ in $(\Gamma^{(-)})^*\cup (\Gamma^{(-)})^*Z_0$, $p_K(Q_k,w^R,Q_l)$ equals $\hat{\pi}_k V_{w^R} \hat{\pi}_l^T$, where $V_{w^R} = V_{w_n}V_{w_{n-1}}\cdots V_{w_1}$.
Furthermore, we set $s(j)=\{j\}\in Q_K$ and $T(i)=\{A\in Q_K\mid i\in A\}$.
The desired equality $p_{N}(i,w,j) = \sum_{r\in T(i)}p_{K}(s(j),w^R,r)$ comes from the above definition of $K$.
\end{proofof}

%%%%%%
%%%%%%

\begin{figure}[t]
\centering
\includegraphics*[height=3.9cm]{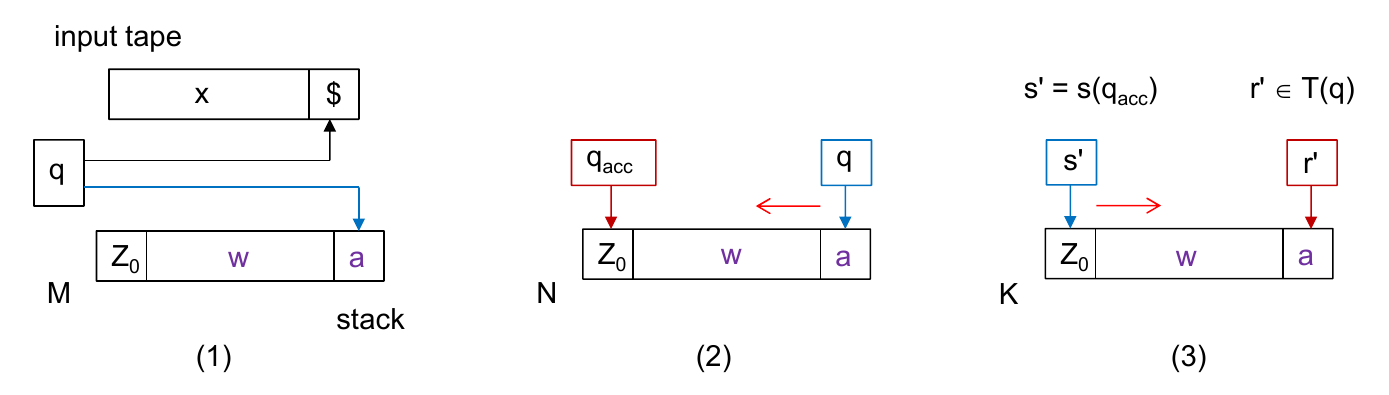}
\caption{An illustration of the relationships among three machines $M$, $N$, and $K$ appearing in the proof of Lemma \ref{final-lambda-move}. (1) The 1ppda $M$ now reaches $\dollar$ and processes it by entering inner state $q$. The stack contains $qwZ_0$. (2) The free 1pfa $N$ takes  its input $awZ_0$ and simulates a final series of $\lambda$-moves of $M$ starting the inner state $q$ and finally enters $q_{acc}$. (3) The free 1pfa $K$ ``inverses'' $N$ starting the inner states $s'$ ($=s(q_{acc})$) and reaching $r'$ ($\in T(q)$).}\label{fig:reversal-1pfa}
\end{figure}

%%%%%
%%%%%

Let us return to Lemma \ref{final-lambda-move} and provide its proof with the help of Lemma \ref{Nasu-Honda}.
We recall that, in the proof of Lemma \ref{remove-left}, we have constructed from an endmarker 1ppda a no-left-endmarker 1ppda  in an ideal shape. This machine thus has the property of making no $\lambda$-moves just after reading the first input symbol.

\vs{-2}
\begin{proofof}{Lemma \ref{final-lambda-move}}
We begin with
a no-left-endmarker 1ppda $M$ of the form $(Q,\Sigma,\{\dollar\}, \Gamma,\Theta_{\Gamma},\delta, q_0, Z_0,Q_{acc},Q_{rej})$ in an ideal shape, equipped with $Q_{acc}=\{q_{acc}\}$ and $Q_{rej}=\{q_{rej}\}$. Our goal is to remove all $\lambda$-moves made by $M$ after reading $\dollar$.
For this purpose, we split $M$'s computation into two parts: (i) the computation of $M$ until it reads $\dollar$ and (ii) the computation of the free 1pfa $N$ associated with a final series of $M$'s $\lambda$-moves after reading $\dollar$.

By Lemma \ref{Nasu-Honda}, we then take the free 1pfa $K=(Q_K,\Gamma^{(-)},\{Z_0\},\delta_K)$ together with three functions $s(\cdot)$,  $T(\cdot)$, and $p_{K}(\cdot)$  that ``simulates'' $N$ on an input by reading its reverse.
Remember that $K$ can mimic a final series of $\lambda$-moves of $M$ (after reading $\dollar$) in the reverse way by an appropriate choice of inner states. Later, nonetheless, we will use $q'_0=s(q_{acc})$ and $\bar{q}'_0 = s(q_{rej})$ as the special inner states for $K$.  Let $Q_{(2)}=\{q'_0,\bar{q}'_0\}$.
If a configuration of $M$ just after reading $\dollar$ is of the form $(q,|x\dollar|,w)$ with $x\in\Sigma^*$, then we assign the probability $p_{K}(q'_0,w^R,T(q))$  ($=\sum_{r\in T(q)} p_K(q'_0,w^R,r)$) to this configuration so that the ``acceptance''  probability of $M$ passing through $(q,|x\dollar|,w)$ is calculated as the  probability of reaching $(q,|x\dollar|,w)$ multiplied by $p_{K}(q'_0,w^R,T(q))$. Similarly, using $p_{K}(\bar{q}'_0,w^R,T(q))$, we can obtain the rejection probability of $M$ passing through $(q,|x\dollar|,w)$.
See Fig.~\ref{fig:reversal-1pfa} for the relationships among $M$, $N$, and $K$.
Since all computation paths of $M$ are assumed to terminate, nevertheless, it suffices for us to  focus only on accepting computation paths of $M$.

Since $M$ empties the stack when it halts, the ideal shape conditions guarantee that  $M$ pops a  top stack symbol at reading $\dollar$ unless the stack is already empty.
To make our description of the desired 1ppda $P$ simpler,  we first eliminate (\#) $M$'s specific transitions of popping a topmost stack symbol at reading $\dollar$ (although this violates the ideal shape conditions). This is achieved by modifying $M$'s set $Q$ of inner states and probabilistic transition function $\delta$ into $\bar{Q}$ and $\bar{\delta}$ as follows. We set $\bar{Q} =Q\cup (Q\times \Gamma)$ and set
$\bar{\delta}(q,\dollar,a\mmid (p,a),a)=\delta(q,\dollar,a\mmid p,\lambda)$ and $\bar{\delta}((p,a),\lambda,a\mmid p,\lambda)=1$. Note that $|\bar{Q}|=n+nm = n(m+1)$.   For convenience,  in the following argument,  we use the same notation $M$ with $Q$ and $\delta$ to describe this modified machine with $\bar{Q}$ and $\bar{\delta}$. To avoid any confusion, we express $n(m+1)$ as $\bar{n}$ so that $|Q|=\bar{n}$ follows.

We introduce a new notation to describe a probability associated with a sub-computation of $M$. For $p,q\in Q$, $x\in\Sigma^*\cup\Sigma^*\dollar$, $a\in\Gamma$, and $w,u\in\Gamma^*$, the notation $\mu_{M}(q,x,w\:\|\: p,u)$ denotes the probability that $M$ starts in inner state $q$ with stack content $w$, reads input $x$ completely, and reaches inner state $p$ with stack content $u$.
It is important to remember that, since $x$ is completely read, this value also includes a (possible)  series of $\lambda$-moves made right after reading the last symbol of $x$ until the stack content becomes $u$ in inner state $p$.
Similarly, $\mu_{P}$ is defined using the machine $P$.

Let $P=(Q',\Sigma,\{\dollar\}, \Gamma',\Theta_{\Gamma'}, \delta',q'_0,Z_0,Q_{acc},Q_{rej})$ denote the desired no-left-endmarker 1ppda, which simulates $M$ as well as $K$.
The machine $P$ needs to remember both inner states of $M$ and $K$ and an input symbol being processed by $K$ during the intended simulation.
When $M$ makes a move by pushing a stack symbol, say, $b$, $P$ simulates this move and, simultaneously, $P$ runs $K$ on this pushed symbol $b$. Since $K$ uses a different set of inner states, we need to remember the changes of such inner states in its own stack. An inner state of $P$ has the form $(q,r_0)$, indicating that $M$ is in inner state $q$ and $K$ begins with inner state $r_0$.
A stack symbol of $P$ has the form $[r,a,r']$, indicating that $K$ is in inner state $r$ reading $a$ and then enters inner state $r'$.

Formally, we set $\Gamma'= [Q_K,\Gamma,Q_{K}]$, $Q'=(Q\times Q_{(2)})\cup (Q\times Q_K\times Q_{(2)})\cup \{q_0\}$, $Q'_{acc}=\{(p,r,q'_0) \mid p\in Q, r\in Q_K\}$, and $Q'_{rej}=\{(p,r,\bar{q}'_0) \mid p\in Q, r\in Q_K\}$.
It then follows that $|Q'| = 2\bar{n}+2\bar{n}^2\cdot 2^{\bar{n}}+1
= 2\bar{n}(\bar{n}2^{\bar{n}}+1)+1$ and $|\Gamma'|= |\Gamma| |Q_K|^2 = m^2 2^{\bar{n}}$.
The desired probabilistic transition function $\delta'$ is constructed to satisfy the equations given in Lines 1--7. Meanwhile, we assume that an input is not $\lambda$.
We wish to define $\delta'$ in Stages (1)--(4). In what follows, let $p,q\in Q$, $p',q'\in Q_K$, $a,b\in\Gamma$, $\sigma\in \Sigma_{\lambda}$, and  $r_0\in Q_{(2)}$.
Moreover, we write $\Gamma_{Z_0}$ for the set $\{[r_0,Z_0,r']\mid r_0\in Q_{(2)}, r'\in Q_K\}$.

%%%%

\s

(1) Let us consider the first step of $M$ in the initial state $q_0$ with the empty  stack. Remember that the first step of $M$ cannot be a $\lambda$-move because $M$ cannot pop $Z_0$.
In the rest of our description, we concentrate on the inner state $q'_0$ ($=s(q_{acc})$),\footnote{It may be possible to use $\bar{q}'_0$ instead of $q'_0$ but, in Stage (4), we need to swap between $q'_0$ and $\bar{q}'_0$.}  which is associated with the acceptance probability of $M$.
If $M$ starts in the initial state $q_0$, reads $\sigma$, and enters inner state $p$ by preserving $Z_0$  (Line 1) or by pushing $b$ (Line 2), then $P$ enters $(p,q'_0)$ from $q_0$ by pushing $[q'_0,Z_0,r']$ (Line 1) or $[r',b,r''][q'_0,Z_0,r']$ (Line 2). These transitions are taken with the corresponding probability of $M$ multiplied by the probability of $K$'s processing $Z_0$ (Line 1) or $bZ_0$ (Line 2).
In what follows, $P$ treats $[q'_0,Z_0,r']$ as if it is a new bottom marker.  

\begin{enumerate}\vs{-1}
  \setlength{\topsep}{-2mm}%
  \setlength{\itemsep}{0mm}%
  \setlength{\parskip}{0cm}%

\item[1.] $\delta'(q_0,\sigma,Z_0 \mmid (p,q'_0),[q'_0,Z_0,r']Z_0) = \delta(q_0,\sigma,Z_0 \mmid p,Z_0)\delta_K(q'_0,Z_0 \mmid r')$.

\item[2.] $\delta'(q_0,\sigma,Z_0 \mmid (p,q'_0),[r',b,r''][q'_0,Z_0,r']Z_0) = \delta(q_0,\sigma,Z_0 \mmid p,bZ_0) \delta_{K}(q'_0,Z_0\mmid r') \delta_K(r',b\mmid r'')$.
\end{enumerate}\vs{-1}

The value $\delta'[q_0,\sigma,Z_0]$ equals the sum of $\sum_{p\in Q} \sum_{r'\in Q_K} \delta'(q_0,\sigma,Z_0\mmid (p,q'_0),[q'_0,Z_0,r']Z_0)$ and $\sum_{(p,b)}  \sum_{r'\in Q_K} \delta'(q_0,\sigma,Z_0\mmid (p,q'_0),[q'_0,b,r']Z_0)$, where the first summation is taken over all pairs $(p,b)$ in $Q\times\Gamma$.
The first term is further calculated by Line 1 as $\sum_{p\in Q}\sum_{r'\in Q_K}\delta(q_0,\sigma,Z_0\mmid p,Z_0) \delta_K(q'_0,Z_0\mmid r') = \sum_{p\in Q}\delta(q_0,\sigma,Z_0\mmid p,Z_0) \cdot \sum_{r'\in Q_K}\delta_K(q'_0,Z_0\mmid r') = \sum_{p\in Q}\delta(q,\sigma,Z_0\mmid p,Z_0) \delta_K[q'_0,Z_0] = \sum_{p\in Q}\delta(q_0,\sigma,Z_0\mmid p,Z_0)$ because of  $\delta_K[q'_0,Z_0]=1$. The second term is similarly calculated according to Line 2. Thus, we obtain $\delta'[q_0,\sigma,Z_0] = \delta[q_0,\sigma,Z_0]$.

\s

(2) Assuming $\sigma\neq\lambda$ and $a\in\Gamma$, let us consider the case where $M$ is in inner state $q$ reading input symbol $\sigma$ and stack symbol $a$ and $P$ is in inner state $(q,q'_0)$ reading stack symbol $[r,a,r']$.
If $M$ pushes $b$ into the stack and $K$ is in inner state $r$ reading $a$, then $P$ enters $(p,q'_0)$ and pushes $[r',b,r'']$ with the same probability multiplied by the extra probability $\delta_{K}(r',b\mmid r'')$ (Line 3)
because $K$ needs to process $b$ if $b$ remains in the stack when $M$ reaches $\dollar$.
It is important to ``remember''  two inner states $r'$ and $r''$ of $K$ in the stack because we need to cancel out this extra probability when $b$ is popped by $M$ in Stage (3). If $M$ does not change the topmost stack symbol $a$, then $P$ only changes inner state  $(q,q'_0)$ to $(p,q'_0)$  (Line 4) since we do not need to simulate $K$'s steps.

\begin{enumerate}\vs{-1}
  \setlength{\topsep}{-2mm}%
  \setlength{\itemsep}{0mm}%
  \setlength{\parskip}{0cm}%

\item[3.] $\delta'((q,q'_0),\sigma,[r,a,r'] \mmid (p,q'_0),[r',b,r''][r,a,r']) =  \delta(q,\sigma,a \mmid p,ba) \delta_K(r',b \mmid r'')$.

\item[4.] $\delta'((q,q'_0),\sigma,[r,a,r'] \mmid (p,q'_0),[r,a,r']) =  \delta(q,\sigma,a \mmid p,a)$.
\end{enumerate}\vs{-1}

Let us calculate the sum $\sum_{(p,b)}  \sum_{r''\in Q_K} \delta'((q,q'_0),\sigma,[r,a,r']\mmid (p,q'_0),[r',b,r''][r,a,r'])$, where the first summation is taken over all pairs $(p,b)$ in $Q\times\Gamma$.  This value is further equal to
$\sum_{(p,b)} \sum_{r''\in Q_K} \delta(q,\sigma,a\mmid p,ba)  \delta_K(r',b\mmid r'') = \sum_{(p,b)} \delta(q,\sigma,a\mmid p,b)\cdot \sum_{r''\in Q_K} \delta_K(r',b\mmid r'') = \sum_{(p,b)} \delta(q,\sigma,a\mmid p,ba) \delta_K[r',b] = \sum_{(p,b)} \delta(q,\sigma,a\mmid p,ba)$ because of $\delta_K[r',b]=1$.
Another sum $\sum_{p\in Q} \delta'((q,q'_0),\sigma,[r,a,r']\mmid (p,q'_0),[r,a,r'])$ is also equal to $\sum_{p\in Q} \delta(q,\sigma,a\mmid p,a)$.
To calculate $\delta'[(q,q'_0),\sigma,[r,a,r']]$, we further need Stage (3). 

\s

(3)
In contrast to Stage (2), when $M$ pops $a$ ($\in\Gamma^{(-)}$), $a$ does not remain in $M$'s  stack while making a final series of $\lambda$-moves. However,  during the early simulation steps,  we have already included into $P$ the probability of $K$'s processing $a$ when $a$ was pushed into the stack by $M$.  We now need to ``cancel out''  such a probability. In other words,  we attempt to nullify the early action of ``$K$'s processing $a$''.
This is a crucial case in our simulation. If $K$ is a reversible machine, we can easily reverse the computation to cancel out the effect of $\delta_{K}$ caused by the previous step. Since $K$ is generally not a reversible machine, nonetheless, we wish to retrieve the previous inner states $r$ and $r'$ from the stack of $P$ and force $K$ to ``roll back'' to the previous situation by summing all probabilities of $K$ associated with $(r,r')$.  This goal is achieved simply as follows: whenever $M$ pops symbol $a$, we make $P$ pop $[r,a,r']$ with the same probability (Line 5).
Let $\sigma\in\Sigma_{\lambda}$.

\begin{enumerate}\vs{-1}
\item[5.] $\delta'((q,q'_0),\sigma,[r,a,r'] \mmid (p,q'_0),\lambda) =  \delta(q,\sigma,a \mmid p,\lambda)$.
\end{enumerate}\vs{-1}

Obviously, $\sum_{p\in Q} \delta'((q,q'_0),\sigma,[r,a,r']\mmid (p,q'_0),\lambda)$ equals $\sum_{p\in Q} \delta(q,\sigma,a\mmid p,\lambda)$. 
Combining Stages  (2) and (3) therefore helps us conclude that $\delta'[(q,q'_0),\sigma,[r,a,r']] = \delta[q,\sigma,a]$.

We briefly explain why  Line 5 actually cancels out the probability $\delta_K(r,a\mmid r')$ that has been already included by Line 3 in a certain early step.
When Line 5 is executed, $M$'s process can be expressed as $(q,i,aw)\vdash_{M}^{>0} (p,i+1,w)$ for a suitable number $i\in\nat^{+}$, provided that $\sigma\neq \lambda$,  and $N$ changes $((q,q'_0),i,[r,a,r']w')$ to $((p,q'_0),i+1,w')$ for a suitable string $w'\in (\Gamma^{(-)})^*Z_0$.
We remark that, when $\sigma=\lambda$, a similar argument works for this case as well.
The probability $\gamma(r')$ of a computation path leading to $((q,q'_0),i,[r,a,r']w')$ from $(q_0,0,Z_0)$ must contain $\delta_K(r,a\mmid r')$ as its multiplicative factor, which is introduced at the time when $[r,a,r']$ has been pushed.
Hence, $\gamma(r')$ can be expressed as $\gamma_0\cdot \delta_K(r,a\mmid r')$ for a certain real value $\gamma_0$, independent of $r'$.
The probability $\gamma$ of leading to $((p,q'_0),i+1,w')$ is the sum of
$\gamma(r') \cdot  \delta'((q,q'_0),\sigma, [r,a,r']\mmid (p,q'_0),\lambda)$ over all $r'\in Q_K$ because the information on $r'$ is completely removed, and therefore all computation paths passing through $((q,q'_0),i,[r,a,r']w')$ must reach the same configuration $((p,q'_0),i+1,w')$.
This value  $\gamma$  equals $\sum_{r'\in Q_K} \gamma(r') \cdot \delta(q,\sigma,a\mmid p,\lambda)$, which is $\gamma_0\cdot \sum_{r'\in Q_K} \delta(q,\sigma,a\mmid p,\lambda)\delta_{K}(r,a\mmid r') = \gamma_0 \: \delta(q,\sigma,a\mmid p,\lambda)
\cdot \sum_{r'\in Q_K}\delta_{K}(r,a,\mmid r') = \gamma_0\: \delta(q,\sigma,a\mmid p,\lambda)
\delta_{K}[r,a]$.  Since  $\delta_{K}[r,a]=1$,  $\gamma$ equals $\gamma_0 \: \delta(q,\sigma,a\mmid p,\lambda)$. This means that $\gamma$ does not contain $\delta_K(r,a\mmid r')$ as its multiplicative factor.
Later, we will formalize this argument and show how this works.

\s

(4)
Assume that  $M$ has already produced a stack content of the form $wZ_0$ and that $M$ processes  $\dollar$. By the elimination of the aforementioned property (\#) of $M$, $M$ conducts neither a push operation nor a pop operation at reading $\dollar$. Hence, $M$ must preserve the top stack symbol.
Notice that, at this moment,  $P$ must have already finished the simulation of $K$ on the input $(wZ_0)^R$. Thus, what is left undone is to decide whether each computation path of $P$ should accept or reject the input.
Let us recall that $K$ correctly provides the acceptance probability of $M$'s making a (possible) final series of consecutive $\lambda$-moves starting in inner state $p$ after reading off $\dollar$ if $K$ starts in inner states in $s(q_{acc})$ and terminates in any inner states $r'$ in $T(p)$.
Remember that the information on $r'$ is already stored in $P$'s stack by one of Lines 1--3. Hence, the acceptance or the rejection of each computation path of $P$ depends on whether $r'$ is in $T(p)$ or not (Lines 6--7).

\begin{enumerate}\vs{-1}
  \setlength{\topsep}{-2mm}%
  \setlength{\itemsep}{0mm}%
  \setlength{\parskip}{0cm}%

\item[6.] $\delta'((q,q'_0),\dollar,[r,a,r'] \mmid (p,r',q'_0),[r,a,r']) =\delta(q,\dollar,a \mmid p,a)$ if $r'\in T(p)$.

\item[7.] $\delta'((q,q'_0),\dollar,[r,a,r'] \mmid (p,r',\bar{q}'_0),[r,a,r']) =\delta(q,\dollar,a \mmid p,a)$ if $r'\notin T(p)$.
\end{enumerate}\vs{-1}

Since $(p,r',r_0)$ for $r_0\in Q_{(2)}$ is a halting state of $P$, Lines 6--7 guarantee that  $P$ does not make any $\lambda$-move after reading $\dollar$. 

\s

By the above definition of $\delta'$ together with $\Gamma_{Z_0}$, we remark that $P$ is in an almost ideal shape.  Next, we verify that $P$ correctly simulates $M$ with the same error probability.
Note that the first transitions of $M$ have the forms:  $\delta(q_0,x_{(1)},Z_0\mmid q_1,Z_0)$ (Line 1) and $\delta(q_0,x_{(1)},Z_0\mmid q_1,bZ_0)$ (Line 2). In what follows, we focus on the transition of $\delta(q_0,x_{(1)},Z_0\mmid q_1,Z_0)$.
Given a series $\boldvec{r}=(r_0,r_1,\ldots,r_{k+1})$ of $Q$ with $r_0=q'_0$ and a string $w=w_kx_{k-1}\cdots w_1w_0$ in $(\Gamma^{(-)})^*Z_0$,  we define $\bar{p}_K(\boldvec{r},w) =\prod_{i=0}^{k}\delta_K(r_i,w_i,r_{i+1})$ and $C(\boldvec{r},w) = [r_k,w_k,r_{k+1}] [r_{k-1},w_{k-1},r_k]
\cdots [r_1,w_1,r_2] [r_0,w_0,r_1]Z_0$ with $w_0=Z_0$. Recall the notation of $\mu_{M}$ and $\mu_{P}$.
Let us consider the following condition.
\begin{quote}
Condition (*): $\mu_{P}((q_1,r_0),x,[r_0,Z_0,r_1]Z_0\:\|\: (p,r_0),C(\boldvec{r},w)) = \mu_{M}(q_1,x,Z_0\:\|\: p,w)\cdot \bar{p}_K(\boldvec{r},w)$ holds for any $\boldvec{r}=(r_0,r_1,\ldots,r_{k+1})$ with $r_0=q'_0$ and a string $w=w_kw_{k-1}\cdots w_1w_0$.
\end{quote}

Henceforth, we wish to verify the correctness of Condition (*) by induction on the length of input string $x$ in $\Sigma^*\dollar$.
The value $\mu_{M}(q_1,x\sigma,Z_0\:\|\: p,bw)$ is calculated from the following four additive terms:
($a$) $\mu_{M}(q_1,x,Z_0\:\|\: q,bw)\delta(q,\sigma,b\mmid p,b)$,
($b$) $\mu_{M}(q_1,x,Z_0\:\|\: q,au)\delta(q,\sigma,a\mmid p,bau)$ with $w=au$,
($c$) $\mu_{M}(q_1,x,Z_0\:\|\: q,abw)\delta(q,\sigma,a\mmid p,\lambda)$, and 
($d$) $\mu_{M}(q_1,x\sigma,Z_0\:\|\: q,abw)\delta(q,\lambda,a\mmid p,\lambda)$.
Notice that, by applying the term ($d$), Concerning the term ($d$), $M$ can pop more than one symbol.
To simplify later descriptions, we abbreviate the tuple $((q_1,r_0),y,[r_0,Z_0,r_1]Z_0)$ as $H_{y}$ for any string $y$. Let us consider the case of $y=x\sigma$.

Let us first consider the term ($b$), which incurs a push operation of $N$ by Line 3.
It then follows by Line 3 that the value $\gamma^{(b)}_{x\sigma} = \mu_P(H_x \:\|\: (p,r_0),C(\boldvec{r}_{+},bau))$, which is described by the term
$\mu_{P}(H_{x}\:\|\: (q,r_0),C(\boldvec{r},au)) \delta'((q,r_0),\sigma,[r_k,a,r_{k+1}]\mmid (p,r_0),[r_{k+1},b,r_{k+2}][r_k,a,r_{k+1}])$, where $\boldvec{r}_{+} = (r_0,r_1,\ldots,r_{k+1},r_{k+2})$.
This term is further calculated as $\mu_P(H_x  \:\|\: (q,r_0),C(\boldvec{r},au)) \delta(q,\sigma,a\mmid p,ba) \delta_K(r_{k+1},b,r_{k+2})$.
By induction hypothesis, we obtain $\mu_{P}(H_x\:\|\: (q,r_0),C(\boldvec{r},au)) = \mu_{M}(q_1,x,Z_0\:\|\: p,au) \bar{p}_K(\boldvec{r},au)$. Hence, it follows that $\gamma^{(b)}_{x\sigma} = \mu_{M}(q_1,x,Z_0\:\|\; p,au) \delta(q,\sigma,a\mmid p,ba)\bar{p}_{K}(\boldvec{r},au)  \delta_K(r_{k+1},b,r_{k+2}) = \mu_M(q_1,x\sigma,Z_0 \:\|\: p,bau) \bar{p}_K(\boldvec{r}_{+},bau)$ since $\bar{p}_K(\boldvec{r}_{+},bau) = \bar{p}_{K}(\boldvec{r},au) \delta_K(r_{k+1},b,r_{k+2})$.

We next deal with the term ($c$), indicating a pop operation given by Line 5. Let us  calculate the  value $\gamma^{(c)}_{x\sigma} = \sum_{r_{k+2}\in Q_K} [ \mu_P(H_x\:\|\; (q,r_0), C(\boldvec{r}_{+},abw)) \delta'((q,r_0),\sigma,[r_{k+1},a,r_{k+2}]\mmid p,\lambda) ]$, where
$\boldvec{r}_{+}$ indicates  $(r_0,r_1,\ldots,r_{k+1},r_{k+2})$.   As mentioned in Stage (3), the summation over all $r_{k+2}$ is needed because the information on $r_{k+2}$ in $\gamma^{(c)}_{x\sigma}$ is lost when $P$ pops  $[r_{k+1},a,r_{k+2}]$.
By the definition of $\delta'$, we further simplify $\gamma^{(c)}_{x\sigma}$ as $\sum_{r_{k+2}\in Q_K}[ \mu_P(H_x\:\|\: (q,r_0), C(\boldvec{r}_{+},abw)) \delta(q,\sigma,a\mmid p,\lambda)]$.
By induction hypothesis follows $\mu_P(H_x\:\|\; (q,r_0), C(\boldvec{r}_{+},abw)) = \mu_M(q_1,x,Z_0\:\|\: q,abw) \bar{p}_K(\boldvec{r}_{+},abw)$. Since $\bar{p}_K(\boldvec{r}_{+},abw) = \bar{p}_{K}(\boldvec{r},bw) \delta_K(r_{k+1},a,r_{k+2})$,  $\gamma^{(c)}_{x\sigma}$ coincides with $\mu_{M}(q_1,x,Z_0\:\|\; q,abw) \delta(q,\sigma,a\mmid p,\lambda) \bar{p}_{K}(\boldvec{r},bw) \cdot \sum_{r_{k+2}\in Q_K} \delta_K(r_{k+1},a,r_{k+2})$. This term further equals $\mu_{M}(q_1,x\sigma,Z_0\:\|\; p,bw)  \bar{p}_{K}(\boldvec{r},bw)$ because $\sum_{r_{k+2}\in Q_K} \delta_K(r_{k+1},a,r_{k+2}) = \delta_K[r_{k+1},a]=1$.

The term ($d$) requires a special attention. We need to treat it by induction on the size of a series of consecutive $\lambda$-moves. We first argue the base case of the first $\lambda$-move.
We focus on the last two steps of $P$ and consider the value $\gamma^{(d)}_{x\sigma}$ defined by $\sum_{r_{k+2},r_{k+3} \in Q_K} [\mu_P(H_x \:\|\; (q,r_0),C(\boldvec{r}_{++},a'a''bw))  \delta'((q,r_0),\sigma,[r_{k+1},a',r_{k+2}]\mmid q',\lambda) \delta'((q',r_0),\lambda, [r_{k+2},a'',r_{k+3}] \mmid p,\lambda) ]$, where  $\boldvec{r}_{++}$ indicates  $(r_0,r_1,\ldots,r_{k+1},r_{k+2},r_{k+3})$.
Using the above argument on the term ($c$), $\gamma^{(d)}_{x\sigma}$ coincides with  $\sum_{r_{k+2},r_{k+3}\in Q} [ \mu_{M}(q_1,x,Z_0  \:\|\;  q',a''bw) \bar{p}_K(\boldvec{r}_{++},a''bw) \delta'((q',r_0),\lambda, [r_{k+2},a'',r_{k+3}]\mmid p,\lambda) ]$.
We further describe it  as the sum of the multiplications of $\mu_{M}(q_1,x,Z_0  \:\|\;  q',a''bw) \delta(q',\sigma,a'\mmid q'',\lambda) \bar{p}_K(\boldvec{r}_{++},a''bw)$ and $\delta(q'',\lambda,a''\mmid p,\lambda)$ over all $r_{k+2},r_{k+3}\in Q$.
Note that $\bar{p}_K(\boldvec{r}_{++},a'a''bw) = \bar{p}_K(\boldvec{r},bw) \delta_K(r_{k+1},a',r_{k+2}) \delta_K(r_{k+2},a'',r_{k+3})$.
Hence, $\gamma^{(d)}_{x\sigma}$ equals $\mu_M(q_1,x\sigma,Z_0   \:\|\;  q',a''bw) \delta(q'',\lambda,a''\mmid p,\lambda) \bar{p}_K(\boldvec{r},bw) \cdot \sum_{r_{k+2},r_{k+3}\in Q} \delta_K(r_{k+1},a',r_{k+2}) \delta_K(r_{k+2},a'',r_{k+3})$.
We remark that $\sum_{r_{k+2},r_{k+3}\in Q_K} \delta_K(r_{k+1},a',r_{k+2}) \delta_K(r_{k+2},a'',r_{k+3})=1$.
It thus follows that $\gamma^{(d)}_x$ equals $\mu_M(q_1,x\sigma,Z_0 \:\|\; p,bw) \bar{p}_K(\boldvec{r}_+,bw)$.
Next, we discuss the case of the second or later $\lambda$-move. The argument for this case is similar to the base case if we replace $\sigma$ by $\lambda$ and use induction hypothesis instead of the result of the term ($c$).

Similarly, we can handle the other term ($a$). We then combine all the terms ($a$)--($d$) to conclude that $\mu_{P}(H_x\:\|\: (p,r_0),C(\boldvec{r},w))$ matches $\mu_M(q_1,x,Z_0\:\|\: p,w)\bar{p}_K(\boldvec{r},w)$. Condition (*) is therefore true.

We have just completed the case of $\delta(q_0,x_{(1)},Z_0 \mmid q_1,Z_0)$.
A similar argument works in the case of $\delta(q_0,x_{(1)},Z_0 \mmid q_1,bZ_0)$, however, we need to pay a special attention to the following situation: a stack content of $N$ becomes $[r_1,a,r_2]Z_0$ during a computation, $N$ then pops $[r_1,a,r_2]$ (to empty the stack), and the next move is made by Line 1 or Line 2 again.

Overall, the acceptance (resp., rejection) probability of $N$ equals the sum of the probability of $M$ reaching a configuration of the form $(q,|x\dollar|,w)$ multiplied by the probability of $K$'s processing $w$ from $q'_0$ (resp., $\bar{q}'_0$).
This completes the proof of the lemma.
\end{proofof}

%%%
\subsection{Removing the Right Endmarker}\label{sec:remove-endmarker}

Throughout Sections \ref{sec:remove-left}--\ref{sec:remove-empty}, we have constructed 1ppda's in almost ideal shapes, which neither have the left endmarker $\cent$ nor make any $\lambda$-move after reading the right endmarker $\dollar$. Let $M$ denote  any of those 1ppda's equipped with only the right endmarker $\dollar$ in an almost ideal shape.
In this final stage of our construction toward Proposition \ref{deleting-endmarker}, we further remove $\dollar$ from $M$.
Notice that, without the endmarker $\dollar$, the 1ppda's cannot in general empty their stacks before halting.

\begin{lemma}\label{no-all-endmarker}
Let $M$ denote any $n$-state no-left-marker 1ppda in an almost ideal shape with stack alphabet size $m$ and assume that there is no $\lambda$-move after reading $\dollar$. Moreover, we assume that $M$ does not alter the stack content at  reading $\dollar$. There exists an error-equivalent no-endmarker 1ppda $N$ using $2n+1$ states, stack alphabet size $(n+1)m$, and push size $3$.
\end{lemma}

We briefly remark on our proof strategy for Lemma \ref{no-all-endmarker}.
When $N$ is reading an input symbol, we first simulate the corresponding move of $M$. We then predict  an encounter of $\dollar$ at the next  move and, based on that prediction, we  simulate the last transition of $M$ made at reading $\dollar$. In the next step, if the following input symbol is indeed $\dollar$ as we have predicted, then we immediately decide either ``accept'' or ``reject'' according to the current inner state. Otherwise, since our prediction is incorrect, we need to probabilistically ``nullify'' the simulation of the last transition of $M$ on $\dollar$ and then probabilistically simulate the next move of $M$. We continue this simulation until $M$ truly  encounters $\dollar$.
As described above, if we try to probabilistically ``rewind'' the simulation of the last transition of $M$ on $\dollar$. For this purpose, if we take a simple strategy of nondeterministically choosing the previous inner state and reversing the last transition of $M$, then the resulting error probability may differ from the original one. Thus, we cannot take this simple strategy in the following proof. Our actual proof is more involved.

%%%

\vs{-2}
\begin{proofof}{Lemma \ref{no-all-endmarker}}
Let $M=(Q,\Sigma,\{\dollar\}, \Gamma,\Theta_{\Gamma}, \delta,q_0, Z_0,Q_{acc},Q_{rej})$ denote any 1ppda in an almost ideal shape with no left endmarker. We assume that, for all inputs, $M$ makes no $\lambda$-move right after reading $\dollar$ and, more importantly, $M$ does not alter any stack content at reading $\dollar$.
Note that, when reading $\dollar$ at the final move, $M$ must make the final decision of whether it accepts or rejects an input without changing the stack content.
This  last transition must be of the form $\delta(p,\dollar,a \mmid r,a)$ for a certain triplet $(p,r,a)$ and we conveniently call this value $\delta(p,\dollar,a \mmid r,a)$ the \emph{decision probability} at $(p,r,a)$.
Notice that $\sum_{r\in Q}\delta(p,\dollar,a \mmid r,a)=1$ holds because $M$ does not make any $\lambda$-move in inner state $p$ at reading $\dollar$.
To eliminate this final move, we must generate the transition probability of this particular move \emph{without reading $\dollar$}.
By the almost ideal shape conditions, we also assume that $\Gamma$ contains a special set $\Gamma_{Z_0}$ whose elements work as if they are bottom markers.  
Our goal is to define a new no-endmarker machine $N=(Q',\Sigma,\Gamma',\Theta_{\Gamma'}, \delta',q'_0,Z_0, Q'_{acc},Q'_{rej})$ that simulates $M$ with the same error probability. Remember that we do not demand $N$ to be in an almost ideal shape.
Hereafter, we explain how to construct the desired machine $N$.
Concerning the key elements $Q'_{acc}$ and $Q'_{rej}$ of $N$,  if $\lambda\in L(M)$, then
we set $Q'_{acc}= Q_{acc}\cup\{q_0\}$ and $Q'_{rej}=Q_{rej}$; otherwise, we set $Q'_{acc}= Q_{acc}$ and $Q'_{rej}=Q_{rej}\cup\{q_0\}$.
We further define $Q'= Q\cup \{p^{(+)}\mid p\in Q\}$ and $\Gamma' = \Gamma\cup  [Q,\Gamma]$.
It then follows that $|Q'|=2n+1$ and $|\Gamma'|=m+nm = m(n+1)$.

\s

(1)
To simulate the first step of $M$, predicting that the next step is the final step, we produce the decision probability associated with this final step. Since this prediction may be false, we need to remember the predicted final move using the stack.
At every step, we cancel out the decision probability taken at the previous step by simply ``ignoring'' the predicated final step (i.e., summing up all the possible final steps) to nullify the previously-taken final step.  After nullifying the supposed final step, we apply a new probabilistic transition and produce the next decision probability. To recover the previous decision probability, since the computation is not in general reversible, we need to remember a triplet $(p,r,a)$. This process is formally expressed as follows. Let $\sigma\in\Sigma$ and $p,q\in Q$.

\begin{enumerate}\vs{-1}
  \setlength{\topsep}{-2mm}%
  \setlength{\itemsep}{0mm}%
  \setlength{\parskip}{0cm}%

\item[1.] $\delta'(q_0,\sigma,Z_0 \mmid p,[r,\bar{Z}_0]Z_0) = \delta(q_0,\sigma,\bar{Z}_0Z_0 \mmid p,\bar{Z}_0) \delta(p,\dollar,\bar{Z}_0 \mmid r,\bar{Z}_0)$ if $\bar{Z}_0\in \Gamma_{Z_0}$.

\item[2.] $\delta'(q_0,\sigma,Z_0 \mmid p,[r,a]\bar{Z}_0Z_0) = \delta(q_0,\sigma,Z_0 \mmid p,a\bar{Z}_0Z_0) \delta(p,\dollar,a \mmid r,a)$ if  $a\in\Gamma^{(-)}-\Gamma_{Z_0}$.
\end{enumerate}\vs{-1}

For an inner state $p$ that $M$ has entered from inner state $q$ in a single step, $N$'s topmost stack symbol $[r,a]$ indicates that the last applied transition is of the form $\delta(p,\dollar,a \mmid r,a)$.
For instance, from Line 1, it follows that $\sum_{r\in Q}\delta'(q_0,\sigma,Z_0 \mmid p,[r,\bar{Z}_0]Z_0) = \sum_{r\in Q} \delta(q_0,\sigma,Z_0\mmid p,\bar{Z}_0Z_0)\delta(p,\dollar,\bar{Z}_0\mmid r,\bar{Z}_0) = \delta(q_0,\sigma,Z_0\mmid p,\bar{Z}_0Z_0)\cdot \sum_{r\in Q}\delta(p,\dollar,\bar{Z}_0\mmid r,\bar{Z}_0) = \delta(q_0,\sigma,Z_0\mmid p,\bar{Z}_0Z_0)$.
The last equality comes from the fact that $\sum_{r\in Q}\delta(p,\dollar,\bar{Z}_0\mmid r,\bar{Z}_0) = \delta[p,\dollar,\bar{Z}_0]=1$.
Therefore, whenever we take all inner states $r$ of $N$ into consideration, $N$ can ``cancel out''  the probability $\delta(p,\dollar,Z_0\mmid r,\bar{Z}_0Z_0)$ from its computation. A similar result follows from Line 2. We generalize this argument as follows.

Let us recall from the proof of Lemma \ref{final-lambda-move} the notation $\mu_M(q,x,w\:\|\: p,u)$, which indicates the probability that $M$ starts in inner state $q$ with stack content $w$, reads off input $x$, and reaches inner state $p$ with stack content $u$.
In a similar way, we further define $\mu_{N}$.
We now want to enforce the following condition between $\mu_{M}$ and $\mu_{N}$.
 \begin{quote}
Condition (*): $\mu_{N}(q_0,x,Z_0\:\|\: p,[r,a]w) = \mu_{M}(q_0,x,Z_0\:\|\: p,aw) \delta(p,\dollar,a \mmid r,a)$ holds for an appropriately chosen quintuple $(p,x,r,a,w)$ with $x\neq\lambda$.
\end{quote}
From Condition (*), it follows that $\mu_{M}(q_0,x\dollar,Z_0\:\|\:  p,aw) = \mu_M(q_0,x,Z_0\:\|\: p,aw) \cdot \sum_{r\in Q} \delta(p,\dollar,a\mmid r,a) = \sum_{r\in Q}\mu_{N}(q_0,x,Z_0\:\|\: p,[r,a]w)$ since $\sum_{r\in Q} \delta(p,\dollar,a\mmid r,a)=\delta[p,\dollar,a] = 1$.
This equation makes it possible to ``nullify'' the previously taken probability $\delta(p,\dollar,a\mmid r,a)$ out of the computation of $N$ by simply  ``ignoring'' the choice of $r$ made at the previous step (or equivalently, taking all choices at once).

Assuming that Condition (*) is true, when $N$ reads off the entire input $x$, $N$ accepts and rejects $x$ with the same error probability as $M$ does.
More precisely, if $M$ halts with probability $\gamma$ by reaching inner state $r$ with stack content $aw$, then $\gamma$ equals $\mu_{M}(q_0,x,Z_0\:\|\:p,aw)\delta(p,\dollar,a \mmid r,a)$, and thus
$N$ halts in inner state $p$ with probability $\gamma$.

In Lines 3--6 shown below, we follow the convention that, when either $\delta(q,\dollar,a \mmid r,a)=0$ or $\delta(q,\dollar,Z_0 \mmid r,Z_0)=0$, the left sides of the corresponding equations must take value $0$. Let $\sigma\in\Sigma_{\lambda}$, $a,b\in\Gamma$,
$p,q\in Q$, and $w\in\Gamma^*$.

\s

(2)
If $M$ pushes $b$ onto topmost symbol $a$ of the stack, then we replace $[r,a]$ by $[s,b]a$ (Line 3). In contrast, if $M$ preserves $a$, then $N$ keeps $[r,a]$ and follows $M$'s move (Line 4).

\begin{enumerate}\vs{-1}
  \setlength{\topsep}{-2mm}%
  \setlength{\itemsep}{0mm}%
  \setlength{\parskip}{0cm}%

\item[3.] $\delta'(q,\sigma,[r,a] \mmid p,[s,b]a) = \delta(q,\sigma,a \mmid p,ba) \delta(p,\dollar,b \mmid s,b)$ if $b\in\Gamma^{(-)}-\{\bar{Z}_0\}$ and $\sigma\neq\lambda$.

\item[4.] $\delta'(q,\sigma,[r,a] \mmid p,[r,a]) =  \delta(q,\sigma,a \mmid p,a)$ if $a\in\Gamma^{(-)}$ and $\sigma\neq\lambda$.
\end{enumerate}\vs{-1}

\s

(3) When $M$ pops $a$, $N$ ``cancels out'' the previous decision probability, pops $[r,a]$ instead (Line 5), replaces any new topmost symbol, say, $b$ by another symbol $[s,b]$, and multiplies the corresponding decision probability (Line 6). Let $\sigma\in\Sigma_{\lambda}$.

\begin{enumerate}\vs{-1}
  \setlength{\topsep}{-2mm}%
  \setlength{\itemsep}{0mm}%
  \setlength{\parskip}{0cm}%

\item[5.] $\delta'(q,\sigma,[r,a] \mmid p^{(+)},\lambda) = \delta(q,\sigma,a \mmid p,\lambda)$ if $a\in\Gamma^{(-)}-\{\bar{Z}_0\}$.

\item[6.] $\delta'(p^{(+)},\lambda,b \mmid p,[s,b]) = \delta(p,\dollar,b \mmid s,b)$.
\end{enumerate}\vs{-1}

For all other cases of $(q,\sigma,d,p,e)$ not listed above, we automatically set $\delta'(q,\sigma,d \mmid p,e)$ to be $0$.

It is important to note that  $\delta'(q,\sigma,[r,\bar{Z}_0] \mmid p,\lambda)=0$ holds for any triplet $(q,\sigma,p)$ because $\bar{Z}_0$ plays as a bottom marker.
Moreover, Condition (*) implies that $\delta'$ correctly simulates $\delta$ until $\dollar$ is read. From this fact, we instantly obtain the lemma.

For the completion of the proof, nevertheless, we still need to verify Condition (*). In an argument similar to the one in the proof of Lemma \ref{final-lambda-move}, we wish to prove Condition (*) by induction on the length of nonempty input $x$.
When $|x|=1$, Lines 1--2 and 5--6 with $\sigma=\lambda$ clearly imply Condition (*). Next, let us consider the case where our input is of the form $x\sigma$ for
$\sigma\in\Sigma$.
In this case, the value $\mu_{M}(q_0,x\sigma,Z_0\:\|\:p,bw)$ with $b\in\Gamma^{(-)}$ is the sum of the following three terms over all $q\in Q$: 
($a$) $\mu_{M}(q_0,x,Z_0\:\|\:q,bw)\delta(q,\sigma,b \mmid p,b)$,
($b$)  $\mu_{M}(q_0,x,Z_0\|q,au)\delta(q,\sigma,a \mmid p,bau)$ with $w=au$, 
($c$)  $\mu_{M}(q_0,x,Z_0\:\|\:q,abw) \delta(q,\sigma,a \mmid p,\lambda)$,  and
($d$)  $\mu_{M}(q_0,x,Z_0\:\|\:q,abw) \delta(q,\lambda,a \mmid p,\lambda)$.

To calculate the value $\mu_{N}(q_0,x\sigma,Z_0\:\|\:p,[s,b]w)$, we first focus on the term ($b$) given by Line 3. We wish to compute the value $\gamma^{(b)}_{x\sigma} = \mu_N(q_0,x\sigma,Z_0 \:\|\: p,[s,b]w)$  induced by the term ($b$) as $\sum_{r\in Q_K}[ \mu_{N}(q_0,x,Z_0\:\|\:q,[r,a]u)\delta'(q,\sigma,[r,a] \mmid p,[s,b]a) ]$ with $w=au$.
Note that the summation over all $r\in Q$ is necessary because the information on $r$ is completely lost after replacing $[r,a]$ by $[s,b]$.
By induction hypothesis, we obtain
$\mu_{N}(q_0,x,Z_0\:\|\:q,[r,a]u) = \mu_{M}(a_0,x,Z_0\:\|\:q,au)\delta(q,\dollar,a \mmid r,a)$.
From this, it follows that $\gamma^{(b)}_{x\sigma}$ equals $\sum_{r\in Q} [\mu_{M}(q_0,x,Z_0\:\|\:q,au) \delta(q,\dollar,a \mmid r,a) \delta(q,\sigma,a \mmid p,ba) \delta(p,\dollar,b \mmid s,b)]$, which further coincides with
$\mu_{M}(q_0,x\sigma,Z_0\:\|\:p,bau)  \delta(p,\dollar,b \mmid s,b) \cdot \sum_{r\in Q}\delta(q,\dollar,a \mmid r,a)$.
We thus conclude that, since $\sum_{r\in Q}\delta(q,\dollar,a\mmid r,a)=1$, $\gamma^{(b)}_{x\sigma}$ equals $\mu_{M}(q_0,x\sigma,Z_0\:\|\:p,bau) \delta(p,\dollar,b \mmid s,b)$.

Next, we concentrate on the term ($c$) given by Lines 5--6.
In this case, we target the value $\gamma^{(c)}_{x\sigma} = \mu_N(q_0,x\sigma,Z_0 \:\|\: p,[s,b]w)$ whose last two steps are described by the term ($c$) as $\sum_{r\in Q}\mu_N(q_0,x,Z_0\:\|\: q,[r,a]bw) \delta'(q,\sigma,[r,a]\mmid p^{(+)},\lambda) \delta'(p^{(+)},\lambda,b\mmid p,[s,b])$.
The necessity of the summation over all $r\in Q$ comes from the fact that $[r,a]$ is popped and lost forever.
The term $\delta'(q,\sigma,[r,a]\mmid p^{(+)},\lambda) \delta'(p^{(+)},\lambda,b\mmid p,[s,b])$ is reduced to $\delta(q,\sigma,a\mmid p,\lambda) \delta(p,\dollar,b\mmid s,b)$.
This value $\gamma^{(c)}_{x\sigma}$ is therefore calculated by induction hypothesis to be $\sum_{r\in Q} \mu_M(q_0,x,Z_0\:\|\: q,abw)\delta(q,\dollar,a\mmid r,a) \delta(q,\sigma,a\mmid p,\lambda) \delta(p,\dollar,b\mmid s,b)$, which further equals  $\mu_M(q_0,x\sigma,Z_0\:\|\: p,bw)\delta(p,\dollar,b\mmid s,b)\cdot \sum_{r\in Q}\delta(q,\dollar,a\mmid r,a)$. Since $\sum_{r\in Q}\delta(q,\dollar,a\mmid r,a) = \delta[q,\dollar,a]=1$, the value $\gamma^{(c)}_{x\sigma}$ becomes $\mu_M(q_0,x\sigma,Z_0\:\|\: p,bw) \delta(p,\dollar,b\mmid s,b)$.

We next look into the term ($d$) induced by Lines 5--6.
We deal with this case by induction on the size of a series of consecutive $\lambda$-moves. As the base case, we focus on the first $\lambda$-move. Consider the value $\gamma^{(d)}_{x\sigma}$,  which is the sum of the multiplications of two terms over all $r_1,r_2\in Q$:
$\mu_N(q_0,x,Z_0 \:\|\: q_1,[r_1,a']a''bw) \delta'(q_1,\sigma,[r_1,a']\mmid q_2^{(+)},\lambda) \delta'(q_2^{(+)},\lambda,a''\mmid q_2,[r_2,a''])$ and $\delta'(q_2,\lambda,[r_2,a'']\mmid q_3^{(+)},\lambda) \delta'(q_3^{(+)},\lambda,b\mmid p,[s,b])$.
By adapting the case of the term ($c$), the first term $\sum_{r_1\in Q} \mu_N(q_0,x,Z_0 \:\|\:  q_1,[r_1,a']a''bw) \delta'(q_1,\sigma,[r_1,a']\mmid q_2^{(+)},\lambda) \delta'(q_2^{(+)},\lambda,a''\mmid q_2,[r_2,a''])$ equals $\mu_M(q_0,x,Z_0 \:\|\:  q_1,a'a''bw) \delta(q_1,\sigma,a'\mmid q_2,\lambda) \delta(q_2,\dollar,a''\mmid r_2,a'')$, from which we obtain $\mu_M(q_0,x\sigma,Z_0  \:\|\: q_2,a''bw) \delta(q_2,\dollar,a''\mmid r_2,a'')$.
The second term $\delta'(q_2,\lambda,[r_2,a'']\mmid q_3^{(+)},\lambda) \delta'(q_3^{(+)},\lambda,b\mmid p,[s,b])$ also equals $\delta(q_2,\lambda,a''\mmid p,\lambda) \delta(p,\dollar,b\mmid s,b)$.
Hence, we obtain $\gamma^{(d)}_{x\sigma}= \mu_M(q_0,x\sigma,Z_0  \:\|\:  q_2,a''bw) \delta(q_2,\lambda,a''\mmid p,\lambda) \delta(p,\dollar,b\mmid s,b) \cdot \sum_{r_2\in Q}\delta(q_2,\dollar,a''\mmid r_2,a'')$, which is further reduced to
$\mu_M(q_0,x\sigma,Z_0 \:\|\: p,bw) \delta(p,\dollar,b\mmid s,b)$.
A general case is similarly handled by replacing $\sigma$ with $\lambda$ and using induction hypothesis instead of the term ($c$).

The remaining term ($a$) is, in essence, handled similarly. By combining all the terms ($a$)--($d$), we then obtain $\mu_N(q_0,x,Z_0\:\|\: p,[s,b]w) = \mu_N(q_0,x,Z_0\:\|\: p,bw) \delta(p,\dollar,b\mmid s,b)$.
Therefore, Condition (*) is true, as requested.
\end{proofof}

%%%%%
%%%%%
\subsection{Proof of Proposition \ref{deleting-endmarker}}\label{sec:final-proof}

As the final process, let us combine all transformations constructed throughout Sections \ref{sec:remove-left}--\ref{sec:remove-endmarker} together with the help of Section \ref{sec:modifications} toward the complete proof of Proposition \ref{deleting-endmarker}. Notice that Propositions \ref{adding-endmarker} and \ref{deleting-endmarker} together imply our main theorem (Theorem \ref{no-endmarker}).
In this subsection, we intend to complete the proof of Proposition \ref{deleting-endmarker}.

\begin{proofof}{Proposition \ref{deleting-endmarker}}
The proof of the proposition begins by taking an arbitrary 1ppda $M = (Q,\Sigma,\{\cent,\dollar\}, \Gamma, \Theta_{\Gamma}, \delta, q_0,Z_0, Q_{acc},Q_{rej})$  with the two endmarkers. Let  $|Q|=n$, $|\Gamma| =m$, and $M$'s push size be $e$.

As a starter,  we place a restriction on $M$ that $M$ must be in an ideal shape. We then apply Lemma \ref{postpone-halting} and then obtain an endmarker  1ppda $M_1$ in an ideal shape with $Q_1$, $\Gamma_1$, and $e_1$ that halts only on or after reading $\dollar$ and also empties its stack. It also follows that  $|Q_1|=2(n+1)$, $|\Gamma_1|=m$, and $e_1=2$.

Sequentially, we apply Lemmas \ref{remove-left}--\ref{no-all-endmarker}.
By Lemma \ref{remove-left}, we obtain a no-left-endmarker 1ppda $M_2$ in an ideal shape with $Q_3$, $\Gamma_2$, and $e_2$ satisfying that $|Q_2| = |Q_1|(|\Sigma|+1)+1 \leq 2(n+1)(|\Sigma|+2)+1\leq 13|\Sigma|n$, $|\Gamma_2|=|\Gamma_1|=m$, and $e_2=2$.

Lemmas \ref{final-lambda-move}--\ref{Nasu-Honda} introduce another no-left-endmarker 1ppda in an almost ideal shape with $Q_3$, $\Gamma_3$, and $e_3$ but making no final $\lambda$-moves. We set $\bar{n}=|Q_2|(|\Gamma_2|+1) \leq 13|\Sigma|n(m+1) \leq 26|\Sigma|nm$. The set  $Q_3$ satisfies that $|Q_3|\leq 2\bar{n}(\bar{n}2^{\bar{n}}+1)+1\leq 4\bar{n}^22^{\bar{n}}+1 \leq 5\bar{n}^22^{\bar{n}}$, which is bounded by $130|\Sigma|^2n^2m^22^{26|\Sigma|nm}$. For other parameters $\Gamma_3$ and $e_3$, we also obtain $|\Gamma_3|\leq  m^22^{\bar{n}}\leq m^22^{26|\Sigma|nm}$ and $e_3=3$.

Finally, Lemma \ref{no-all-endmarker} provides a no-endmarker 1ppda $M_4$ with $Q_4$, $\Gamma_4$, and $e_4$. The size of  $Q_4$ is estimated as  $|Q_4| =  2|Q_3|+1\leq 260|\Sigma|^2n^2m^22^{26|\Sigma|nm}+1 \leq 261|\Sigma|^2n^2m^22^{26|\Sigma|nm}$.
The set $\Gamma_4$ has size $|\Gamma_4|=|\Gamma_3|(|Q_3|+1)$, which is at most $m^22^{26|\Sigma|nm}\cdot 131|\Sigma|^2n^2m^22^{26|\Sigma|nm} \leq 181 |\Sigma|^2n^2m^42^{52|\Sigma|nm}$. Notice that $e_4=3$. If we write $s$ for $26|\Sigma|nm$, then $|Q_4|\leq s^22^s$ and $|\Gamma_4|\leq s^2m^22^{2s} \leq s^42^{2s}$ follow.

In a general case of an arbitrary endmarker 1ppda $M$, we first apply Lemma \ref{transition-simple}.  The lemma helps us construct from $M$ another endmarker 1ppda $M'$ in an ideal shape with $Q'$, $\Gamma'$, and $e'=2$ satisfying that $|Q'|\leq 32en^2m^2(2m)^{8enm}$ and $|\Gamma'|\leq 9enm(2m)^{8enm}$. We then proceed the aforementioned construction process to obtain the desired no-endmarker 1ppda $M_4$.

In the starter case, we have already obtained  (*) $|Q_4|\leq 261|\Sigma|^2n^2m^22^{26|\Sigma|nm}$ and $|\Gamma_4|\leq 181|\Sigma|^2n^2m^42^{52|\Sigma|nm}$.
Since we apply Lemma \ref{transition-simple}, we need to substitute $32en^2m^2(2m)^{8enm}$  and $9enm(2m)^{8enm}$ for $n$ and $m$ in (*), respectively.
It thus follows that $|Q_4|\leq 262|\Sigma|^2(144e^2n^3m^3(2m)^{16enm})^22^{3744|\Sigma|e^2n^3m^3(2m)^{16enm}}$ and $|\Gamma_4|\leq 131|\Sigma|^2(1296e^3n^4m^4(2m)^{32enm})^2 2^{7488|\Sigma|e^2n^3m^3(2m)^{16enm}}$.
For readability, we introduce the notation $d$  to express the value $144|\Sigma|e^2n^3m^3(2m)^{16enm}$. With this notation, we conclude that $|Q_4|\leq 262d^22^{26d}$ and $|\Gamma_4|\leq d^42^{52d}$.

This completes the proof of the proposition.
\end{proofof}

In the above proof of Proposition \ref{deleting-endmarker}, we have proven in part the following statement as the starter regarding an endmarker 1ppda in an ideal shape.

\begin{corollary}
When an $n$-state endmarker 1ppda is in an ideal shape over input alphabet $\Sigma$ with stack alphabet size $m$, we obtain an error-equivalent no-endmarker 1ppda with at most $s^22^s$ inner states and at most stack alphabet size $s^42^{2s}$, where $s=26|\Sigma|nm$.
\end{corollary}

%%%%%%%%%%%%%%%%%%%%%%%%%%%%%%
%%%%%%%%%%%%%%%%%%%%%%%%%%%%%%
\section{Quick Conclusion and Open Problems}\label{sec:open-problem}

In many textbooks as well as scientific papers, two different formulations have been used  to describe various models of one-way pushdown automata in the past: pushdown automata \emph{with two endmarkers} and those \emph{with no endmarkers}.
These two formulations have been well-known to be ``equivalent'' in recognition power for the deterministic and the nondeterministic models of one-way pushdown automata.
Concerning the probabilistic model,  the past literature \cite{MO98,HS10,Yam17,Yam19}, in particular,  used both formulations to obtain useful structural properties of probabilistic pushdown automata but no proof of the  ``equivalence'' between the two formulations (stated as the \emph{no endmarker theorem}) has been published so far. To fill this void, this paper has given the missing formal proof of converting between the two formalisms and, more importantly, it has provided the first explicit upper bound on the stack-state complexity of transforming one formalism to another for probabilistic pushdown automata.

Among all questions that we have left unsettled in this paper concerning the no endmarker theorem (Theorem \ref{no-endmarker}), we wish to list three interesting and important questions.

\begin{enumerate}\vs{-1}
  \setlength{\topsep}{-2mm}%
  \setlength{\itemsep}{0mm}%
  \setlength{\parskip}{0cm}%

\item It is of importance to determine the stack-state complexity of transforming endmarker 1ppda's to error-equivalent no-endmarker 1ppda's. However, this requirement of ``error-equivalence'' seems to be too restrictive. If we weaken the requirement by allowing small additional errors in the process of transformations (this is allowed for bounded-error 1ppda's), can we significantly reduce the stack-state complexity of transformation?

\item We may ask whether there exist much more concise no-endmarker 1ppda's than what we have constructed in Section \ref{sec:deleting-endmarker}. More specifically, for example, is there any no-endmarker 1ppda that has only polynomially many states and exponential stack alphabet size?

\item As we have noted in Section \ref{sec:intro-no-endmarker}, since the deterministic and the nondeterministic models of pushdown automata can be viewed as special cases of the probabilistic model, Proposition  \ref{deleting-endmarker} also provides an upper bound of the stack-state complexity of transformation for the deterministic and nondeterministic models. However, it seems  likely that a much better upper bound can be achieved on the stack-state complexity of transformation by other proof arguments (based on, e.g., grammars). What is the precise stack-state complexity of transformation for those models?
\end{enumerate}

%%%%%%%%%%%%%%%%%
%%%%%%%%%%%%%%%%%

%%%%%%%%%%%%%%%%%%%%%%%%%
%%%%%%%%%%%%%%%%%%%%%%%%%
%%%%%%%%%%%%%%%%%%%%%%%%%
\let\oldbibliography\thebibliography
\renewcommand{\thebibliography}[1]{%
  \oldbibliography{#1}%
  \setlength{\itemsep}{-1pt}%
}
\bibliographystyle{alpha}

\begin{thebibliography}{99}
{\small

\bibitem{EIM19}
J. Eremondi, O. H. Ibarra, and I. McQuillan.
Insertion operations on deterministic reversal-bounded counter machines. Journal of Computer and System Sciences 104 (2019) 244--257.

\bibitem{GG66}
S. Ginsburg and S. Greibach. Determinsitic context free languages. Information and Control 9 (1966) 620--640.

\bibitem{HU79}
J. E. Hopcroft and J. D. Ullman. Introduction to Automata Theory, Languages, and Computation. Addison-Wesley, 1979.

\bibitem{HS10}
J. Hromkovi\v{c} and G. Schnitger. On probabilistic pushdown automata. Information and Computation 208 (2010) 982--995.

\bibitem{IM16}
O. H. Ibarra and I. McQuillan. The effect of end-markers on counter machines and commutativity. Theoretical Computer Science 627 (2016) 71--81.

\bibitem{KGF97}
J. Ka\c{n}eps, D. Geidmanis, and R. Freivalds. Tally languages accepted by Monte Carlo pushdown automata.
In the Proceedings of the International Workshop on Randomization and Approximation Techniques in Computer Science (RANDOM'97), Lecture Notes in Computer Science, vol. 1269, pp. 187--195, 1997.

\bibitem{LP98}
H. R. Lewis and C. H. Papadimitriou. Elements of the Theory of Computation.   Prentice-Hall, 1st edition (1981) and 2nd edition (1998).

\bibitem{MO98}
I. I. Macarie and M. Ogihara. Properties of probabilistic pushdown automata. Theoretical Computer Science 207 (1998) 117--130.

\bibitem{NH68}
M. Nasu and N. Honda. Fuzzy events realized by finite probabilistic automata. Information and Control 12 (1968) 284--303.

\bibitem{Rei92}
K. Reinhardt. Counting and empty alternating pushdown automata. In J. Dassow, editor, Developments in Theoretical Computer Science: 7th IMYSC, number 6 in Topics in Computer Mathematics, Gordon and Breach Science Publishes S.A.,  pp.123--132, 1992.

\bibitem{PP15}
G. Pighizzini and A. Pisoni. Limited automata and context-free languages. Fundamenta Informaticae 136 (2015) 157--176.

\bibitem{Yam08}
T. Yamakami. Swapping lemmas for regular and context-free languages. Manuscript, arXiv:0808.4122, 2008.

\bibitem{Yam14a}
T. Yamakami. Oracle pushdown automata, nondeterministic reducibilities, and the  hierarchy over the family of context-free languages.
In the Proceedings of the 40th International Conference on Current Trends in Theory and Practice of Computer Science, Lecture Notes in Computer Science,
vol. 8327, pp. 514--525, 2014. A complete version at arXiv:1303.1717.

\bibitem{Yam16}
T. Yamakami. Pseudorandom generators against advised context-free languages. Theoretical Computer Science 613 (2016) 1--27.

\bibitem{Yam17}
T. Yamakami. One-way bounded-error probabilistic pushdown automata and Kolmogorov complexity (preliminary report).
In the Proceedings of the 21st International Conference on Developments in Language Theory (DLT 2017),
Lecture Notes in Computer Science, Springer, vol. 10396, pp. 353--364, 2017.

\bibitem{Yam19}
T. Yamakami. Behavioral strengths and weaknesses of various models of limited automata.
In the Proceedings of the 45th International Conference on Current Trends in Theory and Practice of Computer Science (SOFSEM 2019), Lecture Notes in Computer Science, Springer, vol. 11376, pp. 519--530, 2019. A complete and corrected version is available at arXiv:2111.05000, 2021.

}
\end{thebibliography}

%%%%%%%%%%%%%%%%%%%%%%%%%%
%%%%%%%%%%%%%%%%%%%%%%%%%%
\section*{Appendix}

Lemma \ref{transition-simple} relates to the conversion of any (no-)endmarker 1ppda to another (no-)endmarker 1ppda in an ideal shape. However, the proof of the lemma has been omitted from the main text to promote the reader's readability.
Here, we wish to provide the complete proof of the lemma for the completeness.

\vs{-2}
\begin{proofof}{Lemma \ref{transition-simple}}
Let $M = (Q,\Sigma,\{\cent,\dollar\},\Gamma, \Theta_{\Gamma},\delta,q_0, Z_0,Q_{acc},Q_{rej})$ denote any endmarker 1ppda with  $|Q|=n$, $|\Gamma| =m$, and push size $e$.

%%%

To make our later description simpler for any given 1ppda $M$, we further introduce the succinct notation $\delta^*[q,\lambda,a \mmid p,\lambda,w]$ (where  $q\in Q-Q_{halt}$, $p\in Q$,  $a\in\Gamma$, and $w\in\Gamma^*$) for the total probability of the event that, starting in inner state $q$ with stack content $az$ (for an ``arbitrary'' string $z$ satisfying $az\in (\Gamma^{(-)})^*Z_0$), $M$ makes a (possible) series of consecutive $\lambda$-moves without accessing any symbol in $z$, and eventually reaches inner state $p$ with stack content $wz$.
This notation is formally defined as $\delta^*[q,\lambda,a \mmid p,\lambda,w] = \sum_{t\geq 0}\delta^{(t)}[q,\lambda,a \mmid p,\lambda,w]$, where $\delta^{(t)}$ is introduced in the following inductive way. Let $q\in Q-Q_{halt}$, $a\in\Gamma$, and  $w\in (\Gamma^{(-)})^* \cup (\Gamma^{(-)})^*Z_0$.

\renewcommand{\labelitemi}{$\circ$}
\begin{enumerate}\vs{-1}
  \setlength{\topsep}{-2mm}%
  \setlength{\itemsep}{0mm}%
  \setlength{\parskip}{0cm}%

\item[(a)] $\delta^{(0)}[q,\lambda,a \mmid q,\lambda,a]=1$ and $\delta^{(0)}[q,\lambda,a \mmid p,\lambda,w]=0$ for all pairs  $(p,w)\neq (q,a)$.

\item[(b)] $\delta^{(t+1)}[q,\lambda,a \mmid p,\lambda,\lambda] =0$ and
    $\delta^{(t+1)}[q,\lambda,a \mmid p,\lambda,w] = \sum_{r,b,u,u'} \delta^{(t)}[q,\lambda,a \mmid r,\lambda,bu]\delta(r,\lambda,b \mmid p,u')$ if $ww=u'u\neq\lambda$, where $b\in\Gamma^{(-)}$ and $r\in Q$.  
\end{enumerate}\vs{-1}

\n Notice that, in particular, $\delta^*[q,\lambda,a\mmid p,\lambda,\lambda]=0$ follows.
By our definition of 1ppda's in Section \ref{sec:first-formulation}, all computation paths must terminate eventually on every input $x$. Thus, any series of consecutive $\lambda$-moves makes the stack height increase by no more than $e|Q||\Gamma|$ because, otherwise, the series goes into a loop and thus produces an infinite computation path. It thus follows that (*) $\delta^*[q,\lambda,a \mmid p,\lambda,w]\neq0$ implies $|w|\leq e|Q||\Gamma|$.

%%%

We first apply Lemma \ref{postpone-halting} to obtain another 1ppda $N$ that halts only on or after reading $\dollar$.
Recalling the notation used in the proof of Lemma \ref{postpone-halting},
we notice that the obtained machine $N$ with $Q'$, $\Gamma'$, and $\delta'$ in the proof uses only a distinguished set $Q^{(\dollar)}$ of inner states $\bar{q}^{(\dollar)}$ (for $q\in Q$) after reading $\dollar$ and that, when it halts, its stack becomes empty.
More precisely,  once $N$ reads $\dollar$, it makes only $\lambda$-moves with inner states from this special set $Q^{(\dollar)}$.  Notice that $Q^{(\dollar)}\subseteq Q'$  and
$\Gamma'=\Gamma$.
From the proof, we also recall the two halting states $q'_{acc}$ and $q'_{rej}$.
Using these halting states, we obtain $Q'_{acc}=\{q'_{acc}\}$, $Q'_{rej}=\{q'_{rej}\}$,  and $Q'_{halt} = Q'_{acc}\cup Q'_{rej}$.
Recall that $|Q'| = 2(n+1)$, $|\Gamma'| =|\Gamma|=m$, and $e'=e$. Let  $\Gamma_{\lambda} = \Gamma\cup\{\lambda\}$.

Starting with this newly obtained machine $N$, we perform  a series of conversions in Stages  (1)--(4)  to satisfy the desired ideal shape conditions in Stage (5).
In the process of such conversions, we conveniently denote by  $M_i=(Q_i,\Sigma,\{\cent,\dollar\}, \Gamma_i,\Theta_{\Gamma_i}, \delta_i,q_{i,0},Z_{i,0}, Q_{i,acc},Q_{i,rej})$ the 1ppda with push size $e_i$ obtained at each stage $(i)$ with $i\in[4]$.
We further set $Q_{i,halt} = Q_{i,acc} \cup Q_{i,rej}$.

%%%

\s

(1)
We first convert the above obtained 1ppda $N$ to another error-equivalent 1ppda, say,  $M_1$ whose $\lambda$-moves are restricted only to pop operations; namely, $\delta_1(q,\lambda,a \mmid p,w)=0$ for all elements $p,q\in Q_1$, $a\in\Gamma_1$, and $w\in\Gamma_1^+$.
For this purpose, we need to remove any $\lambda$-move by which $M$ changes topmost stack symbol $a$ to a certain nonempty string $w$. We also need to remove all transitions of the form $\delta(q,\lambda,Z_0 \mmid p,Z_0)$, which violates the requirement of $M_1$ with respect to the pop operations.

Formally, we define $Q_1= Q'$,  $Q_{1,acc} = Q'_{acc}$, $Q_{1,rej}= Q'_{rej}$, and $\Gamma_1 = \Gamma^{(-)}\cup \{Z_{0},Z_{1,0}\}$, where $Z_{1,0}$ is a ``new'' bottom marker not in $\Gamma$.
Moreover, we set $q_{1,0}$ to be $q_0$.
It then follows that $|Q_1|=2(n+1)$ and $|\Gamma_1|=m+1$.
The probabilistic transition function $\delta_1$ is constructed formally in the following substages (i)--(iii). For any value of $\delta_1$ not listed below, it is automatically defined to be $0$.

\s

(i) Recall from Sections \ref{sec:first-formulation}--\ref{sec:second-formulation} that, in this paper, the first step of any 1ppda must be a non-$\lambda$-move. At the first step, $M_1$ changes $Z_{1,0}$ to $uZ_0Z_{1,0}$ (for an appropriate string $u\in(\Gamma^{(-)})^*$) so that, after this step,  $M_1$ simulates $M$ using $Z_0$ as an ``ordinary'' stack symbol but with no access to $Z_{1,0}$ (except for the final step of $M_1$). This process is expressed as follows.

\begin{enumerate}\vs{-1}
  \setlength{\topsep}{-2mm}%
  \setlength{\itemsep}{0mm}%
  \setlength{\parskip}{0cm}%

\item[1.] $\delta_1(q_0,\cent,Z_{1,0} \mmid p,uZ_0Z_{1,0}) =   \delta'(q_0,\cent,Z_0 \mmid p,uZ_0)$ for $u$ satisfying $uZ_0\in\Theta_{\Gamma}$.
\end{enumerate}\vs{-1}

(ii) Assume that $az$ is $N$'s stack content and that $N$ makes a (possible) series of consecutive $\lambda$-moves by which $N$ never accesses any stack symbol in $z$.
Consider the case where $N$ changes $a$ to $bu$ by the end of this series and, at reading symbol $\sigma\in\Sigma\cup\{\dollar\}$, $M$ replaces $b$ with $u'$ at the next step in order to produce $u'u$ ($=w$). In this case, we merge this entire process into one single non-$\lambda$-move. Let $q\in Q'$, $a\in \Gamma$, $p\in Q'$, $\bar{p}^{(\dollar)}\in Q^{(\dollar)}$, and $w\in (\Gamma^{(-)})^*\cup (\Gamma^{(-)})^*Z_0$.

\begin{enumerate}\vs{-1}
  \setlength{\topsep}{-2mm}%
  \setlength{\itemsep}{0mm}%
  \setlength{\parskip}{0cm}%

\item[2.] $\delta_1(q,\sigma,a \mmid p,w) =  \sum_{r,b,u,u'}  (\delta')^*[q,\lambda,a \mmid r,\lambda,bu]  \delta(r,\sigma,b \mmid p,u')$ if $\sigma\in\Sigma$ and $q\notin Q^{(\dollar)}$,  where $w=u'u$, $r\in Q'-Q'_{halt}$, and $b\in\Gamma$.

\item[3.] $\delta_1(q,\dollar,a \mmid \bar{p}^{(\dollar)},w) = \sum_{r,b,u,u'}  (\delta')^*[q,\lambda,a \mmid r,\lambda,bu]  \delta'(r,\dollar,b \mmid \bar{p}^{(\dollar)},u')$,  where $w=u'u$, $r\in Q'-Q'_{halt}$, and $b\in\Gamma$.
\end{enumerate}\vs{-1}

\n Since we obtain $\delta_1(q,\sigma,a\mmid p,\lambda) = \sum_{r,b}  (\delta')^*[q,\lambda,a \mmid r,\lambda,b]  \delta'(r,\sigma,b \mmid p,\lambda)$, any non-$\lambda$-moves of $M_1$ may contain pop operations.

Let us consider the next case where, with a certain probability, from stack content $az$, $N$ produces another stack content $bz$ by the end of the series of $\lambda$-moves and then pops $b$ at the next step without reading any input symbol. In this case, we also merge the entire process into a single $\lambda$-move of pop operation as described below.

\begin{enumerate}\vs{-1}
  \setlength{\topsep}{-2mm}%
  \setlength{\itemsep}{0mm}%
  \setlength{\parskip}{0cm}%

\item[4.] $\delta_1(q,\lambda,a \mmid p,\lambda) =  \sum_{r,b} (\delta')^*[q,\lambda,a \mmid r,\lambda,b] \delta'(r,\lambda,b \mmid p,\lambda)$ for  $q\in Q'-Q'_{halt}$ and $a\in\Gamma^{(-)}$, where $r\in Q'-Q'_{halt}$ and $b\in\Gamma^{(-)}$.
\end{enumerate}\vs{-1}

(iii) Assume that $N$ has already read $\dollar$ but it is still in a non-halting state, say, $\bar{q}^{(\dollar)}$ in $Q^{(\dollar)}$ making a series of consecutive $\lambda$-moves. At the next move, $N$ further makes a $\lambda$-move of pop operation.
Unless $N$ reaches $Z_0$, similarly to Line 4., we merge this entire series of $\lambda$-moves and one pop operation into a single $\lambda$-move of pop operation.

\begin{enumerate}\vs{-1}
  \setlength{\topsep}{-2mm}%
  \setlength{\itemsep}{0mm}%
  \setlength{\parskip}{0cm}%

\item[5.]  $\delta_1(\bar{q}^{(\dollar)},\lambda,a \mmid p,\lambda) =   \sum_{r,b}  (\delta')^*[\bar{q}^{(\dollar)},\lambda,a \mmid r,\lambda,b]  \delta'(r,\lambda,b \mmid p,\lambda)$ for  $a\in\Gamma^{(-)}$, $p\in Q'$, and $\bar{q}^{(\dollar)}\in Q^{(\dollar)}$,
where $r\in Q'-Q'_{halt}$ and $b\in\Gamma^{(-)}$.
\end{enumerate}\vs{-1}

\n Once $M$  enters an inner state $p$ in $Q'_{halt}$ at scanning $Z_0$, we remove $Z_0$ and enter the same halting state.
Recall that $Z_0$ is no longer the bottom marker of $M_1$.

\begin{enumerate}\vs{-1}
  \setlength{\topsep}{-2mm}%
  \setlength{\itemsep}{0mm}%
  \setlength{\parskip}{0cm}%

\item[6.] $\delta_1(\bar{q}^{(\dollar)},\lambda,Z_0 \mmid p,\lambda) =    (\delta')^*[\bar{q}^{(\dollar)},\lambda,Z_0 \mmid p,\lambda,Z_0]$ for $\bar{q}^{(\dollar)}\in Q^{(\dollar)}$ and $p\in Q'_{halt}$.
\end{enumerate}\vs{-1}

\n Substages (i)--(iii) shown above guarantee that $M_1$'s $\lambda$-moves are limited to pop operations.  By the aforementioned statement (*), the push size $e_1$ of $M_1$ is at most $e|Q'||\Gamma|\leq 2e(n+1)m \leq 4enm$.

%%%

(2) We next convert $M_1$ to another error-equivalent 1ppda $M_2$ that conducts only the following types of moves:
($a$) $M_2$ pushes only one symbol without changing the exiting stack content, ($b$) $M$ replaces the topmost stack symbol by a (possibly different) single stack symbol, and
($c$) $M_2$ pops the topmost stack symbol.
We also demand that all $\lambda$-moves of $M_2$ are limited only to either ($b$) or ($c$). From these conditions, we obtain $e_2=2$.

To describe the intended conversion, we here introduce two special notations.
The first notation $[w]$ represents a new stack symbol that encodes  an element $w$ in $\Gamma_1^{\leq e_1}$.
To make $M_2$ remember a topmost stack symbol of $M_1$ in an certain inner state, we introduce  another notation $[q,a]$, which indicates that $M_1$ is in inner state $q$ reading $a$ from the stack.

We set $Q_2=[Q_1,\Gamma_1]$, $\Gamma_2=\{[w]\mid w\in\Gamma_1^{\leq e_1}, \exists b\in\Gamma_1\, \exists u\in\Gamma_1^* \, [buw\in \Theta_{\Gamma_1}]\}\cup \{[a]\mid a\in\Gamma_1\}$.
Let $Z_{2,0}=[Z_{1,0}]$ denote a new bottom marker and let $q_{2,0}= [q_{1,0},Z_{1,0}]$ denote a new initial state. It then follows that $|Q_2|=|Q_1||\Gamma_1|\leq 2(n+1)(m+1) \leq 8nm$ and $|\Gamma_2|\leq e_1|\Theta_{\Gamma_1}|+|\Gamma_1|\leq e_1|\Gamma_1|^{e_1+1}+|\Gamma_1|\leq
4enm(2m)^{4enm+1}+2m \leq 8enm(2m)^{8enm}$ since $|\Theta_{\Gamma_1}|\leq |\Gamma_1|^{\leq e_1} = |\Gamma_1|^{e_1+1}$, $|\Gamma_1|=m+1\leq 2m$, and $e_1\leq 4enm$.

\s

(i) Consider the second case where $M_1$ accesses $(\sigma,a)\in \check{\Sigma}\times\Gamma_1$, substitutes $bu$ (with $b\in\Gamma_1$ and $u\in\Gamma_1^*$) for $a$, and enters inner state $p$. In this case,
$M_2$ accesses $(\sigma,[w])$ with a current topmost stack symbol of the form $[w]\in \Gamma_2$, changes $[w]$ to $[u][w]$ (or $[w]$ if $u=\lambda$), and enters the inner state $[p,b]$.
This process is described as follows. Let $p,q\in Q_1$, $\sigma\in\check{\Sigma}$, $b\in\Gamma_1^{(-)}$, and $u,w\in\Gamma_1^{\leq e_1}$.

\begin{enumerate}\vs{-1}
  \setlength{\topsep}{-2mm}%
  \setlength{\itemsep}{0mm}%
  \setlength{\parskip}{0cm}%

\item[1.] $\delta_2([q,a],\sigma,[w] \mmid  [p,b],[u][w]) = \delta_1(q,\sigma,a \mmid p,bu)$ if $u\neq\lambda$.

\item[2.] $\delta_2([q,a],\sigma,[w] \mmid [p,b],[w])  =\delta_1(q,\sigma,a \mmid p,b)$.
\end{enumerate}\vs{-1}

(ii) Assume that $M_1$ accesses $(\lambda,a)$ and pops $a$ by entering inner state $p$ from $q$. When $w\neq\lambda$, $M_2$ accesses $(\sigma,[bw])$ in the inner state $[q,a]$ and changes $[bw]$ to $[w]$ by entering the inner state $[p,b]$.
The function $\delta_2$ is formally defined as follows. Let $a\in\Gamma_1^{(-)}$.

\begin{enumerate}\vs{-1}
  \setlength{\topsep}{-2mm}%
  \setlength{\itemsep}{0mm}%
  \setlength{\parskip}{0cm}%

\item[3.] $\delta_2([q,a],\lambda,[bw] \mmid [p,b],[w]) =\delta_1(q,\lambda,a \mmid p,\lambda)$ if $w\neq\lambda$.

\item[4.] $\delta_2([q,a],\lambda,[b] \mmid [p,b],\lambda) = \delta_1(q,\lambda,a \mmid p,\lambda)$ if $b\neq Z_{1,0}$.
\end{enumerate}\vs{-1}

\n Notice that Lines 3--4 are legitimate definitions because $\delta_2[[q,a],\lambda,[bw]]$ and $\delta_2[[q,a],\lambda,[b]]$ coincide with $\delta_1[q,\lambda,a]$.

\s

(iii) Finally, we need to deal with the special case of $[q_{1,0},Z_{1,0}]$. Because  $Z_{1,0}$ is the bottom marker of $M_1$, a slightly different treatment is required for $Z_{1,0}$ as shown below. Notice that $\delta_1(q,\sigma,Z_{1,0}\mmid p,wZ_{1,0})=0$ if $(q,\sigma)\neq (q_{1,0},\cent)$.

\begin{enumerate}\vs{-1}
  \setlength{\topsep}{-2mm}%
  \setlength{\itemsep}{0mm}%
  \setlength{\parskip}{0cm}%

\item[5.] $\delta_2([q_0,Z_{1,0}],\cent,[Z_{1,0}]\mmid [u,Z_0][Z_{1,0}]) = \delta_1(q_0,\cent,Z_{1,0}\mmid p,buZ_0Z_{1,0})$.
\end{enumerate}\vs{-1}

\n It then follows from Lines 1--5 that $e_2$ equals $2$.

\s

(3) We further convert $M_2$ to $M_3$ so that $M_3$ satisfies all the conditions of Stage (2) and, moreover, $M_3$ conducts no stack operation that replaces any topmost symbol with a ``different'' single symbol. This is done by remembering the topmost stack symbol without writing it down into the stack. For this purpose, we use as a new inner state a symbol of the form $[q,a]$ (where $q\in Q_2$ and $a\in\Gamma_2$) to indicate that $M_2$ is in inner state $q$, reading $a$ in the stack.
Let $Q_3=[Q_2,\Gamma_2]$, $q_{3,0}=[q_{2,0},Z_{2,0}]$, and $\Gamma_3=\Gamma_2\cup \{Z_{3,0}\}$, where $Z_{3,0}$ is a new bottom marker of $M_3$.
We then obtain $|Q_3|=|Q_2||\Gamma_2| \leq 16en^2m^2(2m)^{8enm}$ and $|\Gamma_3|=|\Gamma_2|+1 \leq 9enm(2m)^{8enm}$ since $|Q_2|\leq 8nm$ and $|\Gamma_2|\leq 8enm(2m)^{8enm}$.

Similarly to the notation $\delta^*[q,\lambda,a\mmid p,\lambda,w]$, we define $\delta_2^*[q,\lambda,a\mmid p,\lambda,w]$ from $\delta_2$ to merge a series of consecutive $\lambda$-moves.
If $M_2$ makes a series of $\lambda$-moves that changes $(q,a)$ to $(r,d)$,  followed by a non-$\lambda$-move of changing $(r,d)$ to $(p,b)$ (resp., $(p,bd)$) while reading symbol $\sigma$, then $M_3$ accesses $(\sigma,c)$ for a certain topmost symbol $c$, and changes its inner state $[q,a]$ to $[p,b]$ without changing $c$ (resp., with changing $c$ to $ac$). In contrast, if $M_2$ changes $(q,a)$ to $(p,b)$ after a series of $\lambda$-moves, then $\delta_3$ satisfies Condition (4) of the ideal shape form. Note that a pop operation of $M_3$ is performed only with a $\lambda$-move.

Formally, $\delta_3$ is constructed as follows. Let $p,q\in Q_2$ and
$a,b\in\Gamma_2$, $c\in\Gamma_3$, and $\sigma\in\check{\Sigma}$.

\begin{enumerate}\vs{-1}
  \setlength{\topsep}{-2mm}%
  \setlength{\itemsep}{0mm}%
  \setlength{\parskip}{0cm}%

\item[1.] $\delta_3([q,a],\sigma,c \mmid [p,b],c)= \sum_{r,d} \delta_2^*[q,\lambda,a \mmid r,\lambda,d] \delta_2(r,\sigma,d \mmid p,b)$.

\item[2.] $\delta_3([q,a],\sigma,c \mmid [p,b],ac) = \sum_{r,d} \delta_2^*[q,\lambda,a \mmid r,\lambda,d] \delta_2(r,\sigma,d \mmid p,bd)$.

\item[3.] $\delta_3([q,a],\lambda,c \mmid [p,c],\lambda) = \sum_{r,d} \delta_2^*[q,\lambda,a \mmid r,\lambda,d] \delta_2(r,\lambda,d \mmid p,\lambda)$.

\item[4.] $\delta_3([q,a],\lambda,c\mmid [p,b],c) = \sum_{r,d}\delta_2^*[q,\lambda,a\mmid r,\lambda,d] \delta_2(r,\lambda,d\mmid p,b)$.
\end{enumerate}\vs{-1}

\n Each of the above summations is taken over all pairs $(r,d)$ in $Q_2\times \Theta_{\Gamma_2}$. Lines 1--4 obviously guarantee the aforementioned requirement for $M_3$. Notice that $e_3=2$.

\s

(4) We convert $M_3$ to $M_4$ to satisfy all the conditions of Stage (3) and to make only $\lambda$-moves of pop operations that follow only a series of pop operations.
Let $Q_4=Q_3\cup\{p'\mid p\in Q_{halt}\}$, $\Gamma_4=\Gamma_3$,  $q_{4,0}=q_{3,0}$, and $Z_{4,0}=Z_{3,0}$.
It follows that $|Q_4|\leq 2|Q_3|\leq 32en^2m^2(2m)^{8enm}$ and $|\Gamma_4|\leq 9enm(2m)^{8enm}$.
The probabilistic transition function $\delta_4$ is constructed as follows. A basic idea of our construction is that, when $M_3$ makes a pop operation after a certain non-pop operation, we combine them as a single move. As a new notation, we set $\hat{p}$ to be $p$ if $p\notin Q_{3,halt}$, and $p'$ otherwise. Let $\sigma\in\check{\Sigma}_{\lambda}$.

\begin{enumerate}\vs{-1}
  \setlength{\topsep}{-2mm}%
  \setlength{\itemsep}{0mm}%
  \setlength{\parskip}{0cm}%

\item[1.] $\delta_4(q,\sigma,a \mmid \hat{p},\lambda) = \sum_{r\in Q_3} \delta_3(q,\sigma,a \mmid r,a) \delta_3(r,\lambda,a \mmid p,\lambda)$.

\item[2.] $\delta_4(q,\sigma,a \mmid \hat{p},a) = \sum_{r\in Q_3} \delta_3(q,\sigma,a \mmid r,ba) \delta_3(r,\lambda,b \mmid p,\lambda)$ if $\sigma\neq\lambda$.

\item[3.] $\delta_4(q,\sigma,a \mmid p,ba) = \delta_3(q,\sigma,a \mmid p,ba)$ if $\sum_{s\in Q_3}\delta_3(p,\lambda,b \mmid s,\lambda)=0$ and $\sigma\neq\lambda$.
\end{enumerate}\vs{-1}
We then obtain $e_4=2$.

\s

(5) Finally, we set the desired 1ppda $N$ to be $M_4$. It is clear from Stages (1)--(4) that $N$ satisfies all the ideal shape conditions. Thus, the proof of the lemma is completed.
\end{proofof}

Before closing this Appendix, we wish to give a short remark.
Stages (3) and (4) in the above proof  can be combined so that we can further remove all $\lambda$-moves by introducing a new set of instructions. More precisely, we set a new  transition function $\delta$ as a map from $(Q-Q_{halt})\times \Sigma\times (\Gamma \cup \PP(\Gamma^{(-)})) \times Q\times \Gamma^*$ to $[0,1]$, where the notation  $\delta(q,\sigma,\{a_1,a_2,\ldots,a_k\}\mmid p,\lambda)$ means that we must repeat the pop operation as long as the topmost stack symbols are in $\{a_1,a_2,\ldots,a_k\}$ whenever possible (otherwise, we abort the computation). For this $\delta$, we need to demand that $\delta(q,\sigma,\{a_1,a_2,\ldots,a_k\} \mmid p,u)>0$ implies $u=\lambda$.

%%%%%%%%%%%%%%%%%%%%%%%%%%%%%%%%%%%%%%%%%%%%%%%%%
%%%%%%%%%%%%%%%%%%%%%%%%%%%%%%%%%%%%%%%%%%%%%%%%%
\end{document}